\author{} % Enter your name
\title{} %Enter the title of your report 
\date{\today} % insert a specific date	

\documentclass[a4paper,11pt]{article}
\usepackage[left=30mm,top=30mm,right=30mm,bottom=30mm]{geometry}
\usepackage{etoolbox} %required for cover page
\usepackage{booktabs}
\usepackage[table,xcdraw]{xcolor}
\usepackage[usestackEOL]{stackengine}
\usepackage[round]{natbib}
\usepackage[T1]{fontenc}
\usepackage[utf8]{inputenc}
\usepackage{bm}
\usepackage{graphicx}
\usepackage{subcaption}
\usepackage{amsmath}
\usepackage{amsfonts}
\usepackage{mathtools}
\usepackage{xcolor}
\usepackage{float}
\usepackage{hyperref}
\usepackage[capitalise]{cleveref}
\usepackage{enumitem,kantlipsum}
\usepackage{amssymb}
\usepackage{amsbsy}
\usepackage{amsthm}
\usepackage{bbm}% theorems, definitions, etc.
\usepackage{pifont}
 \usepackage{multirow}
 \usepackage{footnote}
 \makesavenoteenv{tabular}
 \usepackage{url}
\usepackage{wrapfig,lipsum}
\usepackage{caption} % for '\ContinuedFloat' directive

\usepackage[ruled,vlined]{algorithm2e}
\usepackage{listings}
\usepackage{dirtytalk}
\usepackage{graphicx}
\usepackage{multirow}

\usepackage{chngcntr}
\usepackage{apptools}
\AtAppendix{\counterwithin{lemma}{section}}
\AtAppendix{\counterwithin{theorem}{section}}

\newtheorem*{theorem*}{Theorem}

\newtheorem{definition}{Definition}

\newtheorem{remark}{Remark}
\newtheorem*{remark*}{Remark}
\newtheorem{assumption}{Assumption}
\newtheorem*{assumption*}{Assumption}

  \newtheorem{theorem}{Theorem}
\newtheorem{lemma}{Lemma}

\newtheorem*{notation*}{Notation}

\newcommand{\indep}{\perp \!\!\! \perp}

\hypersetup{
    colorlinks,
    linkcolor={black},
    citecolor={blue!50!black},
    urlcolor={blue!80!black}
}

\graphicspath{{figures/}}	

\providecommand{\keywords}[1]
{
  \small	
  \textbf{\textit{Keywords: }} #1
}
\title{Extreme Treatment Effect: Extrapolating Dose-Response Function Into Extreme Treatment Domain
} 
\author{Juraj Bodik$^{1 \footnote{ Email of the corresponding author: Juraj.Bodik@unil.ch}}$}
\date{%
    $^1$ {\small HEC, University of Lausanne, Switzerland} \\%
    }

\begin{document}

\pagenumbering{gobble}% Remove page numbers (and reset to 1)

\maketitle
\begin{abstract}
The potential outcomes framework serves as a fundamental tool for quantifying causal effects. The average dose–response function (also called the effect curve), denoted as $\mu(t)$, is typically of interest when dealing with a continuous treatment variable (exposure). The focus of this work is to determine the impact of an extreme level of treatment, potentially beyond the range of observed values—that is, estimating $\mu(t)$ for very large $t$. Our approach is grounded in the field of statistics known as extreme value theory. We outline key assumptions for the identifiability of the extreme treatment effect. Additionally, we present a novel and consistent estimation procedure that can potentially reduce the dimension of the confounders to at most $3$. This is a significant result since typically, the estimation of $\mu(t)$ is very challenging due to high-dimensional confounders. In practical applications, our framework proves valuable when assessing the effects of scenarios such as drug overdoses, extreme river discharges, or extremely high temperatures on a variable of interest.
\end{abstract}
%TC:ignore
%TC:endignore
  
\newpage
%\listoffigures
%\newpage
%\listoftables
%\newpage
%\listofalgorithms % List of algorithms in pseudocode format
%\newpage
%\listoflistings % List of algorithms in code format
%\newpage

\pagenumbering{arabic}% Arabic page numbers (and reset to 1)

% This is how you can organize your document
\section{Introduction}

Quantifying causal effects is a fundamental problem in many diverse fields of research \citep{rosenbaum1983central, holland1986statistics, robins2000marginal, Imai_misunderstandings}. Some prevalent examples include the impact of smoking on developing cancer \citep{imai2004causal}, the influence of education on increased wages \citep{heckman2018returns}, the effects of various meteorological factors on precipitation \citep{hannart2018probabilities} or the effect of policy design on various economy factors \citep{low2017use}.  

The potential outcomes framework \citep{rubin2005} has been the fundamental language to express the notion of the causal effect. The crux of this framework lies in acknowledging that, in any given scenario, multiple potential outcomes exist based on different interventions or exposures \citep{Rubin}. This perspective challenges researchers to consider not only the observed outcome but also the unobserved outcomes that could  materialize under alternative conditions. The typical focus in causal inference is on the binary treatment variable (exposure). However, binary treatment is unable to differentiate between different levels of the treatment variable. This issue can be partially solved by assuming a continuous treatment. For example,  \cite{Kennedy_2016} and \cite{westling2020causal} proposed an estimator based on local linear smoothing. \cite{galagate2016causal} discussed the combination of parametric and non-parametric models for the effect curve estimation. 
\cite{rubin2006extending} and \cite{neugebauer2007nonparametric} utilized marginal structural causal models framework. \cite{zhang2023exploring} and\cite{Bica2020_generative_adversarial_networks} applied neural networks for the effect estimation. However, typical methods that work with a continuous treatment variable are not well suited for the inference that goes beyond the observed range of the data.  

In this paper, we are interested in the \textit{extreme treatment effect}; that is, the quantity of interest is the effect of an extreme level of treatment--outside of the observed range. Consider the following example from medicine: assume that the data of a study (either randomized or observational) is available to us, with the health status ($Y$) of patients and the corresponding dose of a medicine administrated ($T$). The data available only depicts $Y$ when $T \leq 20\, mg$. What if then, we would like to know the change in $Y$ when the dose is increased to $T=25mg$? Answering this inquiry is hard, since we have zero data to answer it (this might be considered unethical to give such a dose to a patient), and we must rely on strong unverifiable assumptions and extrapolation. Additionally, in the case of observational studies, high-dimensional confounders pose yet another significant challenge.

The connection between causal inference and extreme value theory has been receiving increasing interest. \cite{zhang2018extremal} and \cite{Deuber} 
analyze the Extreme Quantile Treatment Effect (EQTE) of a binary treatment on a continuous, heavy-tailed outcome. The paper authored by \cite{huang2022extreme} develops a method to estimate the EQTE and the Extreme Average Treatment Effect (EATE) for continuous treatment.  \cite{bodik} developed a framework for Granger-type causality in extremes. Some other approaches for causal discovery using the extreme values include
 \cite{Gnecco, Pasche, krali2023heavytailed, bodik2023structural}. \cite{EngelkeGraphicalModels} propose graphical models in the context of extremes.  \cite{Naveau} analyzed the the effect of climate change on weather extremes. \cite{courgeau2021extreme} proposed a framework for extreme event propagation.  
\cite{10.1214/20-AOAS1355} study probabilities of necessary and sufficient causation as defined in the counterfactual theory using the multivariate generalized Pareto distributions. We contribute to this growing literature and provide a theoretically well-founded approach for estimation and inference of the extremal treatment effect. 

Recent advancements in machine learning research have spotlighted the extrapolation capabilities of different models \citep{dong2023first, christiansen2021causal, saengkyongam2024identifying,chen2021domain}. For instance, 'engression,' as proposed by \cite{XinweiShen2023engression}, presents a framework that serves as an extrapolating alternative to regression-based neural networks. Similarly, \cite{pfister2024extrapolationaware} explored extrapolation of conditional expectations by assuming that the maximum derivative occurs within the observed range of support. While these approaches aren't inherently causal, they can be construed as such under certain assumptions. Despite achieving cutting-edge performance, these methods encounter two primary limitations: difficulty handling multiple confounders and reliance on often uninterpretable extrapolation assumptions. In contrast, our framework focuses on the causal aspect of the extrapolation, and can handle a large number of confounders. This is achieved under weak assumptions commonly embraced in the extreme value theory. Moreover, our framework relies on strong yet more interpretable extrapolation assumptions. While our primary focus is on linear regression, our approach has the flexibility to integrate with various machine learning methodologies, potentially improving overall performance. However, this integration may come at the expense of losing interpretability for certain assumptions.

As for the application in this work, we consider dataset describing extreme precipitation and river discharge levels in Switzerland. A historical record indicates a maximum precipitation level near Zurich's meteo-station on 6.6.2002, reaching an extreme of $111 \frac{mm}{m^2}$. This event coincided with the river Reuss (near Zurich) nearly breaching its banks, causing damage to adjacent settlements. We focus on the following question: how would the river discharge alter if the precipitation on that day were to reach $120 \frac{mm}{m^2}$? Would the river breach its banks under such circumstances? We anticipate that the effect of precipitation on river discharge may vary between the body of the distribution and its tail.  This anticipation stems from several factors: During periods of light to moderate precipitation, the ground absorbs a significant portion of the rainfall, reducing its contribution to the river flow. In contrast, during severe rainstorms, a larger proportion of the precipitation directly contributes to the river flow, potentially resulting in a more pronounced impact on discharge levels. Therefore, we expect to observe differing, potentially more severe, impacts of extreme precipitation on discharge levels compared to moderate events, highlighting the importance of understanding such dynamics across varying levels of precipitation intensity.

The structure of this paper is as follows: we introduce the notation and preliminaries on causal inference and extreme value theory in Section~\ref{section_problem_statement}. In Section~\ref{Section_tail_framework}, we present the main assumptions along with some simple theoretical implications. In Section~\ref{section_algorithm}, we introduce a practical statistical methodology for estimating an extreme treatment effect from data. Section~\ref{section_simulations} explains our methodology using a simple simulated example and discusses simulation results. In Section~\ref{section_application2}, we explore the application of inferring the effect of extreme precipitation on river discharge levels.

This manuscript includes five appendices: Appendix~\ref{section_application} introduces a second real-world application regarding the compressive strength of concrete, which has been relocated to the appendix for the sake of brevity. Appendix~\ref{Appendix_Simulations} contains a detailed simulation study, exploring the methodology under various conditions, including 1) a varying dimension $dim(\textbf{X})$, 2) comparison with classical methods from the literature, 3) a hidden confounder, 4) different dependence structures, and 5) varying dose-response functions. Appendix~\ref{Appendix_application2} contains a detailed inference process for the river-application described in Section~\ref{section_application2}. Appendix~\ref{appendix_consistency} provides a more detailed explanation of the bootstrap algorithm used in the inference process and presents the theory behind the consistency result. Finally, proofs can be found in Appendix~\ref{Section_proofs}.

\section{Problem statement, notation and preliminaries}
\label{section_problem_statement}
\begin{wrapfigure}{r}{3cm}
\includegraphics[scale=0.4]{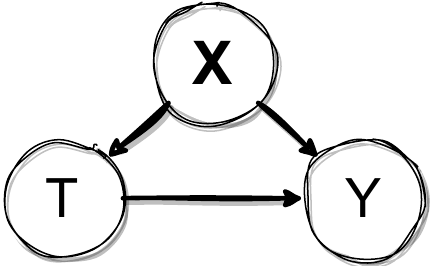}
\label{Simple_graph_confoudner}
\end{wrapfigure} 
Following \cite{Hirano}, we define dose–response functions in the potential outcomes framework. We consider the triplet of $(\textbf{X}, T, Y)$, where $\textbf{X} \in \mathcal{X} \subseteq \mathbb{R}^d$, $T \in \mathcal{T} \subseteq \mathbb{R}$, and $Y \in \mathcal{Y}\subseteq\mathbb{R}$ denote the confounders, treatment, and response variable, respectively, in an observational causal study. We assume a continuous treatment setting, where $\mathcal{T} = (\tau_L, \tau_R)$ for some $\tau_L, \tau_R\in\overline{\mathbb{R}}:=\mathbb{R}\cup \{-\infty\}\cup \{\infty\}$.  Here, $\tau_R\in \mathbb{R}\cup \{\infty\}$ is the right endpoint of the support of $T$. For simplicity, assume that the right endpoint of the support of $T\mid \textbf{X}=\textbf{x}$  is equal to $\tau_R$ for all $\textbf{x}\in\mathcal{X}$. 
Let $Y(t)$ be a set of potential outcomes corresponding to the hypothetical world in which $T = t$ is set deterministically. The fundamental problem of causal inference arise, since in the real world, each individual can only receive one treatment level $T$ and we only observe the corresponding outcome $Y = Y (T)$. 

We observe a random sample $\{\textbf{X}_i, T_i, Y_i\}_{i=1}^n$ of size $n\in\mathbb{N}$. It follows from our setting that, given the observed covariates, the distribution of potential outcome for one unit is assumed to be unaffected by the particular treatment assignment of another unit (Stable Unit Treatment Value Assumption). We utilize the letter $H$ for a (possible) hidden confounder. 
We denote vectors by bold letters. For any pair of continuous random variables $Z$ and $Z'$, we denote its probability distribution function $P_Z(\cdot)$, density function $p_Z(\cdot)$ and a conditional density $p_{Z\mid Z'}(\cdot\mid \cdot)$.

The average dose–response function and patient-specific dose–response function are defined as $$\mu(t)=\mathbb{E}[Y(t)], \,\,\,\,\,\,\,\,\,\mu_\textbf{x}(t)=\mathbb{E}[Y(t)\mid \textbf{X}=\textbf{x}], $$respectively. Although the term ``dose'' is typically  associated with the medical domain, we adopt here the term dose-response learning in its more general setup: estimating the causal effect of a treatment on an outcome across different (continuous) levels of treatment. Our objective is to learn the behaviour of $\mu(t)$ or $\mu_\textbf{x}(t)$ for $t\approx\tau_R$. 

For a pair of real functions $f_1, f_2$, we employ the following notation: $f_1(t)\sim f_2(t)$ for $t\to\tau_R$, if $\lim_{t\to\tau_R}\frac{f_1(t)}{f_2(t)} = 1$. Sequence of random variables approach same distribution as sample size grows. In the remaining of the paper, we assume that $\mu(t), \mu_{\textbf{x}}(t)$ are continuous on some neighbourhood of $\tau_R$ for all $\textbf{x}$.

\subsection{Classical assumptions}\label{section2.1}

Two classical assumptions in the literature for identifying the average dose–response function are:
\begin{itemize}
    \item \textbf{Unconfoundedness}: Given the observed covariates, the distribution of treatment is independent of potential outcome. Formally, we have $T \indep Y(t) \,|\, \textbf{X}, \, \forall t \in \mathcal{T}$, where $\indep$ denotes the independence of random variables. 
    \item \textbf{Positivity}: $p_{T\mid \textbf{X}} (t\mid \textbf{x})>0$ for all $t,\textbf{x}$, where $p_{T\mid \textbf{X}}$ represents the conditional density function of the treatment given the covariates.
\end{itemize}

Under these assumptions, \cite{Hirano} showed the identifiability of the dose-response function via  \begin{equation}\label{eq49}
    \mu_\textbf{x}(t) = \mathbb{E}[Y | T = t, \textbf{X} = \textbf{x}], \,\,\,\, \text{and} \,\,\,\,\, \mu(t) = \mathbb{E}[\mu_{\textbf{X}}(t)] = \mathbb{E}[\mathbb{E}[Y | T = t, \textbf{X} = \textbf{x}]],
\end{equation} 
where the inner expectation is taken over $Y$ and the outer expectation is taken over $\textbf{X}$. 

Even if we are not willing to rely on the unconfoundedness assumption, it may often still be of interest to estimate the function $t\to \mathbb{E}[\mathbb{E}[Y | T = t, \textbf{X} = \textbf{x}]]$ as an adjusted measure of association, defined purely in terms of observed data. It may be interpreted as the average value of $Y$ in a population with exposure fixed at $T=t$ but otherwise characteristic of the study population with respect to $\textbf{X}$ \citep{gill2001causal, Kennedy_2016, westling2020causal}.   

When the positivity assumption is violated, a different type of extrapolation arises \citep{king_2006}, which is different from the one considered in this work. This scenario occurs when the distributional support of variable $T$ varies across different levels of confounding variables $\textbf{X}$. Various approaches have been devised to confront this challenge, including propensity thresholding \citep{crump2009dealing}.

Several algorithms were proposed to estimate the function $\mu(t)$ in the body of the distribution of $T$. State-of-the-art-methods estimate $\mu(t)$ via $\sum_{i: T_i\approx t}w_iY_i$ for appropriate weights $w_i$ what serve to ``erase'' the confounding effect of $\textbf{X}$ \citep{ai2021estimation, Continuous_Treatment_Effect_Estimation_Using_Balancing, Continuous_Treatment_Effect_adversarial_deconfounding, Kennedy_2016, kreif2015evaluation}.
Typically, the estimation of $\mu(t)$ involves a two-step procedure \citep{imai2004causal, Hirano, zhao2020propensity}. In the first step, we model the distribution $T\mid \textbf{X}$, also known as the propensity. In the second step, we model the distribution of $Y\mid T$, suitably adjusted by the propensity, with the aim of mitigating the confounding effect of $\textbf{X}$.

\cite{Hirano} introduced a generalized propensity score (GPS) defined as $e(t,\textbf{x}):=p_{T\mid \textbf{X}}(t\mid \textbf{x})$. One common approach is to model $p_{T\mid \textbf{X}}$ using a Gaussian model. In binary treatment cases (when $\mathcal{T} = \{0,1\}$), the propensity score is a probability denoted as $e(1,\textbf{X}) = P(T=1\mid \textbf{X})$ and is typically modeled using logistic regression. Subsequently, we define weights $w_i$ as $w_i:=\frac{1}{\hat{e}(T_i, \textbf{X}_i)}$ or stabilized weights $w_i:=\frac{\hat{p}_T(T_i)}{e(T_i, \textbf{X}_i)}$, where we additionally model and estimate the marginal distribution of $T$, denoted as $\hat{p}_T$.

In a similar vein, \cite{imai2004causal} introduced the concept of a ``uniquely parameterized propensity function assumption,'' which states that for every value of $\textbf{X}$, there exists a unique finite-dimensional parameter $\theta \in \Theta$ such that $e(\cdot\mid \textbf{X})$ depends on $\textbf{X}$ only through $\theta(\textbf{X})$. Since $\theta(\textbf{X})$ contains all information about the confounding, we only model $\mathbb{E}[Y\mid T=t, \theta(\textbf{X})=\textbf{s}]$ instead of $\mathbb{E}[Y\mid T=t, \textbf{X}=\textbf{x}]$ in equation (\ref{eq49}). In a vast majority of applications, $\theta(\textbf{X})$ corresponds to the parameters of a normal distribution. 

\subsection{Extreme value theory}
\label{Section_preliminaries_EVT}
When dealing with extreme values, it is easy to introduce a strong selection bias. A naive approach for estimating $\mu(t)$ for large values of $t$ might involve only considering observations where $t$ exceeds a certain threshold, denoted as $\tau$, and computing $\mu(t)$ using conventional techniques, while disregarding all values below this threshold.  This is a typical approach of many classical algorithms, which estimate $\mu(t)$ by focusing solely on a local neighborhood of observations around $t$. However, this approach can introduce a significant selection bias. In its extreme manifestation, all observations where $t$ exceeds $\tau$ might exclusively pertain to men, for instance. The selection bias arises if the effect of $T$ on $Y$ differs between men and women (see Figure~\ref{Plot1} in Section~\ref{Simulations_simple_example} with $\tau = 3$). We employ the Extreme Value Theory technique known as peaks-over-threshold to tackle this issue.

Extreme value theory is a sub-field of statistics that explores techniques for extrapolating the behavior (distribution) of $T$ beyond the observed values. A limiting theory posits that the tail of $T$ can be well approximated by the Generalized Pareto Distribution (GPD), as detailed in the following explanation.

Consider a sequence $(T_i)_{i\geq 1}$ of independent and identically distributed (iid) random variables with a common distribution $F$, and $M_n=\max_{i=1, \dots, n} T_i$ represents the running maximum. It is well known \citep{resnick2008extreme} that if (\textbf{Condition 1:}) there exists a non-degenerate distribution $G$ such that $\frac{M_n-b_n}{a_n}\overset{D}{\to}G$ as $n\to\infty$ for some sequences of constants $\{a_n, b_n\}_{n=1}^\infty\in\mathbb{R}_+^\mathbb{N}\times\mathbb{R}^\mathbb{N}$ , then $G$ falls within the Generalized Extreme Value (GEV) distribution family. Condition~1 can equivalently be expressed using the following definition:
\begin{definition}(\cite{iii1975statistical})
    The distribution $F$ is in the max-domain of attraction of a generalized extreme value distribution (notation $F\in MDA(\gamma)$) if there exist $\gamma\in\mathbb{R}$ and sequences of constants $a_n>0, b_n\in\mathbb{R}, n=1,2,\dots$ such that 
    $
    \lim_{n\to\infty} F^n(a_nx + b_n) = exp(- (1+\gamma x)^{-1/\gamma})
    $
for all $x$ satisfying $1+\gamma x>0$. In case $\gamma=0$, the right side is interpreted as $ exp(- e^{-x})$. The parameter $\gamma$ is called the extreme value index (EVI). 
\end{definition}

This condition is mild as it is satisfied for most standard distributions, for example, the normal, Student-t and beta distributions. The following crucial theorem states that the tail of $T$ can be well approximated by GPD if the distribution of $T$ belongs to $MDA(\gamma)$.

\begin{theorem*}[Theorem 4.1 in \cite{Coles}]\label{ThmColes}
Let $T\overset{}{\sim} F\in MDA(\gamma)$. Then, for large $\tau\approx \tau_R$, there exist $\sigma>0, \gamma\in\mathbb{R}$ such that the distribution of  $T-\tau\mid T>\tau$ is approximately $GPD(0, \sigma, \gamma)$. 
\end{theorem*}

GPD distribution has three parameters, namely a location $\tau\in\mathbb{R}$, scale $\sigma>0$ and a shape $\gamma\in\mathbb{R}$. Its distribution function takes a form: 
\begin{equation*}
  H(x) = \begin{cases}
  & 1-\left(1+\gamma\frac{x-\tau}{\sigma}\right)^{-1/\gamma}   \text{,    }\gamma\neq 0,\\ 
 &  1 - \exp\left(-\frac{x-\tau}{\sigma}\right)  \text{,    }\gamma= 0,
\end{cases}
\end{equation*}
defined on the support $[\tau - \sigma/\gamma, \infty), (-\infty, \infty), (-\infty, \tau - \sigma/\gamma]$ for cases $\gamma<0, \gamma=0, \gamma>0$ respectively. Cases when $\gamma>0, \gamma=0$, and $\gamma<0$ correspond to the Fréchet, Gumbel, and Weibull distributions, respectively \citep{fisher_tippett_1928, resnick2008extreme}.  

Note that when the distribution of $T - \tau$ given $T > \tau$ follows a GPD with parameters $0$, $\sigma$, and $\gamma$, an equivalent assertion can be made that $T$ given $T > \tau$ follows a GPD with parameters $\tau$, $\sigma$, and $\gamma$. We denote $\theta = (\tau, \sigma, \gamma)^\top$. 

\begin{assumption}\label{assumption_max_domain}
We assume that the distributions $T$ and $T\mid \textbf{X}$  are in the max-domain of attraction of a generalized extreme value distribution. 
\end{assumption}

\section{Our tail framework}
\label{Section_tail_framework}
We aim to model the effect of a treatment variable $T$ in the context of extreme values of $T$. However, it's essential to approach the term 'extreme' with caution, considering the discrepancy between real-world implications and the interpretations within our model. Take, for instance, if $T$ represents a drug dose in milligrams; our model operates under the assumption that $T$ tends toward $\tau_R$, which can be potentially larger than several kilograms. While this mathematical abstraction lacks practical significance—given that administering several kilograms of a drug is physically implausible—the model does include values of $T$ arbitrarily large. Of course, we do not claim that our model performs well when $T$ equals several kilograms, but only that it performs well for $T$ in the 'reasonable neighborhood' of the observed values, where the effect of $T$ is expected to be 'extrapolatable' from the observed values. 

\subsection{Assumptions}
We are not aiming to estimate the complete $\mu(t)$ but rather only its values for large $t$. Therefore, we can relax the classical assumptions for the identification of $\mu(t)$; what we specifically need is their tail counterparts. 
\begin{assumption}[Unconfoundedness in tail]\label{assumption_unconfoundness_tail} 
For all $\textbf{x}\in\mathcal{X}$ holds
\begin{equation}\label{Unconfoundedness in tail1}\tag{Unconfoundedness in tail}
\mathbb{E}[Y(t)\mid \textbf{X}=\textbf{x}] \sim \mathbb{E}[Y\mid \textbf{X}=\textbf{x}, T=t]\,\,\,\,\,\,\,\,\,\,\,\,\text{ as }t\to\tau_R.
\end{equation}
We always assume the existence of the expected values. 
\end{assumption}
Rather than simply writing $t\to\tau_R$, we frequently opt for the notation $t(\textbf{x})\to\tau_R$ to emphasize its dependence on the random variable $\textbf{X}$. 
Note that Assumption~\ref{assumption_unconfoundness_tail} is strictly less restrictive than the Unconfoundedness assumption introduced in Section~\ref{section2.1}. 
\begin{remark} \label{remark1}
To provide some intuition regarding the permissiveness of Assumption~\ref{assumption_unconfoundness_tail}, we rephrase our framework in the language of structural causal models (SCM, \cite{Pearl_book}). Assume that the data-generating process of the output $Y$ is as follows: $$Y = f_Y(T,\textbf{X}, H, \varepsilon),\,\,\,\,\,\, \varepsilon\indep (T,\textbf{X},H).$$
Here, $H$ represents a (possible) latent confounder of $T$ and $Y$. Then, the dose-response function has a form $\mu(t) = \mathbb{E}[f_Y(t, \textbf{X}, H, \varepsilon)]$ where the expectation is taken with respect to joint $(\textbf{X},H, \varepsilon)$. 

Assumption~\ref{assumption_unconfoundness_tail} can be rephrased as follows: 
There exist a function $\tilde{f}_Y$ such that 
\begin{equation}
    \label{SCM23}
    \tag{Unconfoundedness in tail in SCM}
    f_Y(t, \textbf{x}, h, e) \sim \tilde{f}_Y(t,\textbf{x},e) \text{  as }t\to\tau_R,
\end{equation}
for all admissible values of $\textbf{x},h,e$.  
This assumption is valid for example in additive models; that is, when  $f_Y(t,\textbf{x},h,e) = \tilde{f}_Y(t,\textbf{x},e) + g(h)$ for some functions $\tilde{f},g$. 
\end{remark}

Additionally, we restate the positivity assumption in the context of its tail counterpart.
\begin{assumption}[Positivity in tail]\label{assumption_positivity_tail}  $p_{T\mid \textbf{X}} (t\mid \textbf{x})>0$ for all $\textbf{x}$ and all $t>t_0$ for some $t_0\in\mathcal{T}$, where $p_{T\mid \textbf{X}}$ represents the conditional density function of the treatment given the covariates.
\end{assumption}
Note that this assumption is weaker than Assumption~\ref{assumption_max_domain}.
\subsection{Adjusting only for $\theta(\textbf{X})$ }

The following lemma serves as a tail counterpart of an identifiability for the classical framework. It states that, under Assumptions  \ref{assumption_unconfoundness_tail} and \ref{assumption_positivity_tail},  the tail of the dose-response function is identifiable from the observational distribution via the propensity function $\pi_0(t,\textbf{x})$. 

\begin{lemma}[Identifiability]\label{lemma1}
    Under Assumptions \ref{assumption_unconfoundness_tail} and \ref{assumption_positivity_tail} holds 
    \begin{equation}
         \mu(t)  \sim \mathbb{E}\{\pi_0(T, \textbf{X})Y\mid T=t\}, \,\,\,\,\,\,\,\text{ as } t\to\tau_R,
    \end{equation}
 where $\pi_0(t,x) := \frac{p_T(t)}{p_{T\mid \textbf{X}}(t\mid \textbf{x})}$ is the (stabilized) propensity function.
\end{lemma}

Recall that the distribution of  $T \mid \textbf{X}=\textbf{x}$, conditioned on $T>\tau(\textbf{x})$ for large $\tau(\textbf{x})\approx \tau_R$, is approximately GPD with parameters $\theta(\textbf{x}) = (\tau(\textbf{x}), \sigma(\textbf{x}), \gamma(\textbf{x}))$. The following result suggests that instead of conditioning on (potentially high-dimensional) covariates $\textbf{X}$, we only need to condition on $\theta(\textbf{X})$. 

\begin{lemma}\label{lemma2}
Under Assumptions~\ref{assumption_max_domain} and \ref{assumption_unconfoundness_tail}, for all $\textbf{s}$ in the support of $\theta(\textbf{X})$ holds
\begin{equation}\label{eq1}
  \mathbb{E}[Y(t)\mid  \theta(\textbf{X})=\textbf{s}]\sim \mathbb{E}[Y\mid  T=t, \theta(\textbf{X})=\textbf{s}]\,\,\,\,\,\,\,\,\, \text{   for    }t\to\tau_R.  
\end{equation}
Hence, 
$$\mu(t) \sim \int \mathbb{E}[Y\mid  T=t, \theta(\textbf{x})]p_{\theta(\textbf{X})}(\textbf{x})d\textbf{x}\,\,\,\,\,\,\,\,\, \text{   for    }t\to\tau_R.$$
\end{lemma}
Lemma~\ref{lemma2} suggests that it is sufficient to condition only on $\theta(\textbf{X})$ rather than on $\textbf{X}$. This finding is pivotal for dimension reduction, effectively reducing the dimension from $dim(\textbf{X})$ to at most $3$. Nonetheless,  this is merely a limiting result, and it introduces an approximation error of the GPD approximation for finite samples.

\subsection{Model for the conditional expectation of $Y$ given a $T$}

\begin{figure}[h]
\centering
\includegraphics[scale=0.5]{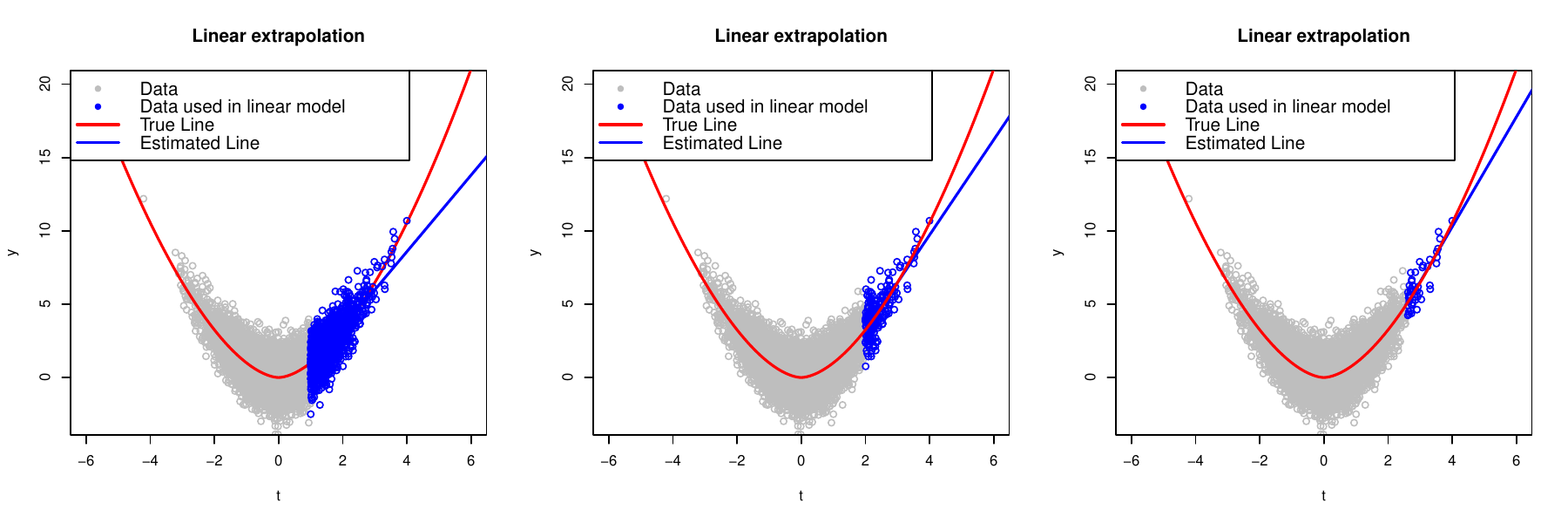}
\includegraphics[scale=0.395]{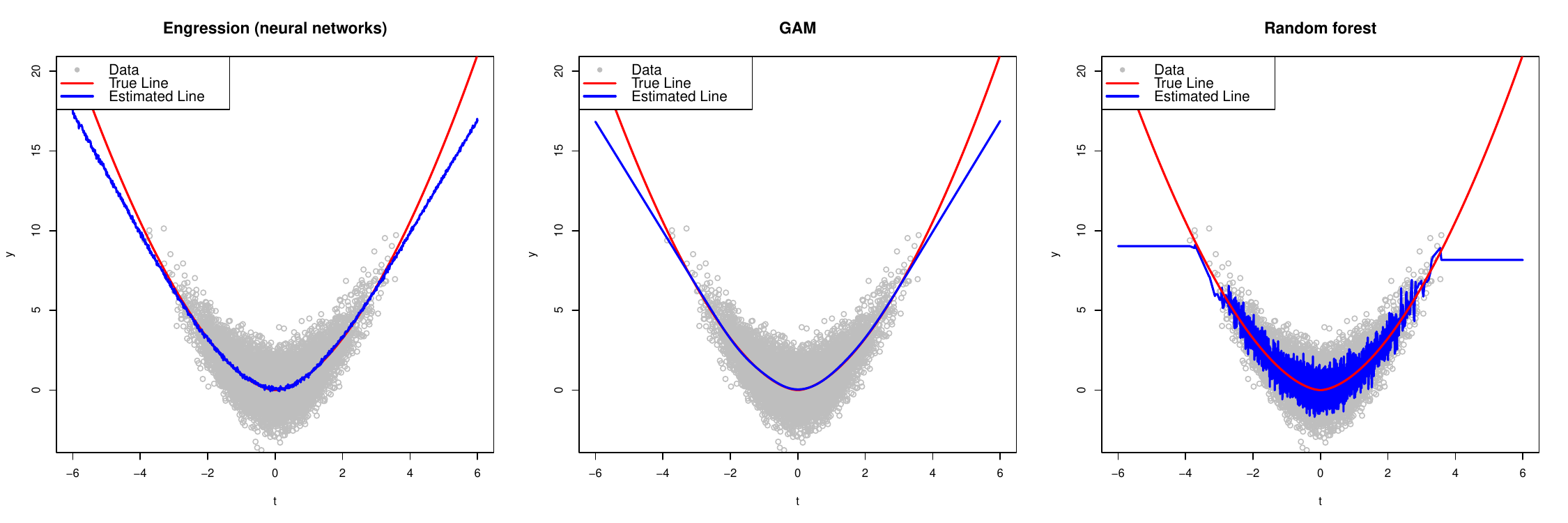}
\caption{Extrapolation of $\mathbb{E}[Y\mid T=t]$ under various models (without confounding). The upper three figures illustrate a first-order extrapolation approach, employing a linear fit at the boundary of the support of $T$. In contrast, the lower three figures depict estimations generated by distinct models: the first utilizes a pre-additive noise model parameterized by neural networks \citep{XinweiShen2023engression}, the second employs smoothing splines \citep{Wood}, and the third utilizes a random forest approach \citep{breiman2001random}.
}
\label{Plot_extrapolation_illustration}
\end{figure}

Under Assumption~\ref{assumption_unconfoundness_tail}, modeling $\mu_{\textbf{X}}$ reduces to a statistical modeling of $\mathbb{E}[Y\mid T, \textbf{X}]$. Furthermore, under Assumptions~\ref{assumption_max_domain} and \ref{assumption_unconfoundness_tail}, it reduces to modeling $\mathbb{E}[Y\mid T, \theta(\textbf{X})]$. In principle, a wide range of models can be considered, ranging from simple linear models to non-parametric neural networks. The principle of Occam's razor suggests that, especially when extrapolating beyond the range of observed values, simpler models often prove to be the most effective choices \citep{soklakov2002occam}. The extrapolation capabilities of various models have recently garnered attention in machine learning research. \citep{XinweiShen2023engression} introduced the 'engression' framework as an extrapolating counterpart to regression-based neural networks. While we build our framework under a linear model for simplicity, it is possible to utilize different models, such as the engression-based ones. Figure~\ref{Plot_extrapolation_illustration} illustrates the extrapolating properties of various commonly used models.

A straightforward approach to model $ \mathbb{E}[Y\mid T=t, \textbf{X}=\textbf{x}]$ under the assumption of linearity-in-the-tail would be assuming an existence of functions $\tilde{\alpha}, \tilde{\beta}$ such that 
\begin{equation}
    \label{eq4232}
    \mathbb{E}[Y\mid T=t, \textbf{X}=\textbf{x}] \sim \tilde{\alpha}(\textbf{x}) + \tilde{\beta}(\textbf{x})t, \,\,\,\,\,\,\text{  as }t\to\tau_R.
\end{equation}
Following the notation in Remark~\ref{remark1}, this corresponds to assuming \begin{equation*}
f_Y(t, \textbf{x}, h, e) \sim \tilde{\alpha}(\textbf{x}) + \tilde{\beta}(\textbf{x})t \text{  as }t\to\tau_R,
\end{equation*}
for all admissible values of $h,e$.  This assumption is valid for example in additive models where $f_Y(t, \textbf{x}, h, e) =\tilde{\alpha}(\textbf{x})+ \tilde{\beta}(\textbf{x}) t + g(\textbf{x},h,e)$ for some function $g$.  
However, using the result from Lemma~\ref{lemma2}, it is sufficient to condition only on $\theta(\textbf{X})$ instead of  potentially high-dimensional $\textbf{X}$. Therefore, we introduce the following model assumption: 
 
\begin{assumption}[Conditional linearity of tail]\label{assumption_linearity_of_tail_1}
There exist functions $\alpha$ and $\beta$ such that for all $\textbf{s}$ in the support of $\theta(\textbf{X})$, the following holds:

\begin{equation}
\label{linear_model}
    \mathbb{E}[Y\mid T=t,\theta(\textbf{X})=\textbf{s}] \sim \alpha(\textbf{s}) + \beta(\textbf{s})t, \,\,\,\,\,\,\text{  as }t\to\tau_R.
\end{equation}

\end{assumption}

Such an assumption was explored in various contexts (typically where $\theta(\textbf{X})$ represents parameters of a normal distribution, see \cite{imai2004causal} or Section 2.2.1 in  \cite{zhao2020propensity}. To the best of our knowledge the extreme case was not yet explored). We can construct our inference method (as discussed in Section~\ref{section_algorithm}) by estimating $\alpha$ and $\beta$ using various machine-learning methodologies.

\section{Inference and estimation}
\label{section_algorithm}

Let  $(\textbf{x}_i, t_i, y_i)_{i=1}^n$ be the observed data. In the following, we propose a methodology for the estimation of $\mu(t)$ for $t\approx \tau_R$ under Assumptions~\ref{assumption_max_domain}, \ref{assumption_unconfoundness_tail} and \ref{assumption_linearity_of_tail_1}.

 Consider the following two-step procedure. In the first step, we approximate the tail of $T\mid \textbf{X}$ using GPD (that is, we estimate the location, scale and shape parameters ${\theta}(\textbf{X}) = (\tau(\textbf{X}), \sigma(\textbf{X}), \xi(\textbf{X}))$). In the second step, we estimate the expectation of $Y$ given a large $T$ conditional on the estimated GPD parameters $\hat{\theta}(\textbf{X})$. 

\begin{enumerate}
    \item Estimate $\theta(\textbf{x})$: \begin{itemize}
    \item Choose $q\in(0,1)$.
    \item Estimate covariant-dependent threshold $\tau(\textbf{x})$ using a quantile regression: That is, estimate $q$-quantile of $T\mid \textbf{X}=\textbf{x}$.
    \item From now on, restrict our inference on the observations from $S:=\{i: t_i>\hat{\tau}(\textbf{x}_i)\}$. 
    \item Estimate $\theta(\textbf{x})$ in the tail-model: That is, estimate $(\sigma, \xi)$  from the data-points in $S$ in the model where  $$T\mid T>\hat{\tau}(\textbf{x}), \textbf{X}=\textbf{x}\sim GPD(\hat{\tau}(\textbf{x}), \sigma(\textbf{x}), \xi(\textbf{x})).$$
\end{itemize}    
\item Estimate $\mu(t)$ or $\mu_{\textbf{x}^\star}(t)$ using $\hat{\theta}(\textbf{x})$:
\begin{itemize}
    \item Estimate $\alpha, \beta$ in model~(\ref{linear_model}) from the data-points in $S$ (that is, we only consider $t>\hat{\tau}(\textbf{x})$). 
    \item Return  $\hat{\mu}(t):=\frac{1}{n}\sum_{i=1}^n  \hat{\alpha}[\hat{\theta}(\textbf{x}_i)]+ \hat{\beta}[\hat{\theta}(\textbf{x}_i)]t$ or $\hat{\mu}_{\textbf{x}^\star}(t):=\hat{\alpha}[\hat{\theta}(\textbf{x}^\star)]+ \hat{\beta}[\hat{\theta}(\textbf{x}^\star)]t$. 
\end{itemize}
\end{enumerate}

The first step is a very standard procedure in extreme value literature called 'peak-over-threshold' \citep{Coles}; it is standard to assume constant shape parameter $\gamma(\textbf{x}) \equiv \gamma\in\mathbb{R}$ \citep{Smith, Davison2} since in practice, it is untypical for the shape parameter to change with covariates. For the estimation of $\tau(\textbf{x}), \sigma(\textbf{x}), \alpha(\textbf{x}), \beta(\textbf{x})$, we use either linear parametrisation \footnote{That is, $\tau(\textbf{x}) = \tau^\top\textbf{x}, \sigma(\textbf{x})=\sigma^\top \textbf{x}, \alpha(\textbf{s}) = \alpha^\top\textbf{s}, \beta(\textbf{s}) = \beta^\top\textbf{s}$ for some real coefficients $\tau, \sigma, \alpha, \beta$ and their estimation is done via classical maximum likelihood. } or non-parametric smooth estimation using splines (GAM, \cite{Wood}), but any method can be used in practice. In case of a very small sample size, we can also assume a constant scale parameter $\sigma(\textbf{x})\equiv \sigma\in\mathbb{R}$ in order to reduce the dimension of the estimation.

The choice of $q$ in the first step is a standard issue in extreme value theory \citep{Schneider2021, Caeiro2015, Davison1990}. For theoretical results, $q$ should be growing with the sample size; that is, $q=q_n$ satisfying $\lim_{n\to\infty}q_n=1$ and $\lim_{n\to\infty}n(1-q_n)=\infty$.   In practical terms, $q$ should be set as high as possible while ensuring that a sufficient amount of data remains above the threshold to maintain good inferential properties. Classical choices include $q=0.9$, $q=0.95$ or $q=0.99$, depending on the size of our dataset.

We utilize the basic bootstrap technique (sometimes also called Efron's percentile method, see Chapter 23 in  \cite{Vaart_1998}) to establish confidence intervals. This involves randomly sampling, with replacement, from our dataset to generate multiple bootstrap samples, each matching the size of our original dataset. For each bootstrap sample, we calculate the estimate of the statistic $\hat{\mu}^\star(t)$. Subsequently, we determine the $\alpha$-percentiles of the re-sampled statistics to derive the confidence intervals. Details can be found in Appendix~\ref{appendix_bootstrap_definition}.

\begin{remark}\label{remark_error}
One must be cautious when interpreting confidence intervals during extrapolation. Generally, estimation of ${\mu}(t)$ is subject to two primary sources of bias: 1) bias stemming from model misspecification, and 2) bias arising from estimation variance. While the former bias can be mitigated within the body of the distribution by comparing different models and employing cross-validation, AIC or BIC criteria to select the most suitable model, this approach becomes less reliable in the extremal region. Eliminating this bias necessitates observation of data within the region of interest. The latter bias stemming from estimation uncertainty can be addressed by computing confidence intervals (in our case, via bootstrap). It's important to note that our bootstrap confidence intervals only account for the latter bias and consequently, the first type of bias presents a greater challenge during extrapolation since it is, in principle, unquantifiable without additional data.
\end{remark}

\begin{theorem}[Idea: Precise statements can be found in Appendix~\ref{appendix_consistency}]\label{Theorem_consistency}
Assuming the conditions outlined in either Theorem~\ref{theorem_consistency2} or Theorem~\ref{theorem_consistency1} are met, our procedure is consistent. Furthermore, under the assumptions detailed in Theorem~\ref{theorem_bootstrap_consistency}, the bootstrap confidence intervals are asymptotically consistent at a correct level.
\end{theorem}

In our procedure, we adopt a practice common in extreme value theory, where we concentrate solely on the extreme observations (set $S$) while discarding all non-extreme values. This approach stems from the rationale that extreme observations provide the most valuable insights into out-of-support behavior. Utilizing data within the body of the distribution may introduce bias, as these values may exhibit different behavioral patterns. Mathematically, this rationale can be expressed through an examination of the precision of the GPD approximation. This approximation shows high precision exclusively in extreme values while displaying bias and low precision for non-extreme values.

\section{Illustration and experiments}
\label{section_simulations}

 In this section, we assess the performance of our methodology using both a simple illustrative example and experimental data. A comprehensive simulation study is provided in Appendix~\ref{Appendix_Simulations}. 

The quantity of interest in the application presented in Section~\ref{section_application2} is the difference $\mu(t_1) - \mu(t_2)$ for $t_1, t_2 < \tau_R$.  Hence, in the simulations, we focus on estimating: 
\[
\omega_\textbf{x} := \lim_{t \to \infty} [\mu_\textbf{x}(t+1) - \mu_\textbf{x}(t)] \quad \text{or} \quad \omega := \lim_{t \to \infty} [\mu(t+1) - \mu(t)],
\]
assuming $\tau_R = \infty$ and that the corresponding limits exist. Note that under the linear model (\ref{linear_model}), the limit exists and corresponds to the parameter $\omega_\textbf{x} = \beta[\theta(\textbf{x})]$.
Here, $\omega_\textbf{x}$ can be regarded as the tail counterpart of a coefficient $\beta_\textbf{x}$ in a linear model $Y = \alpha_\textbf{x} + \beta_\textbf{x}T + \varepsilon$, where $\alpha_\textbf{x}$ and $\beta_\textbf{x}$ are real coefficients, possibly dependent on $\textbf{X}$.

\subsection{Simple example}
\label{Simulations_simple_example}

\begin{figure}[]
\centering
\includegraphics[scale=0.46]{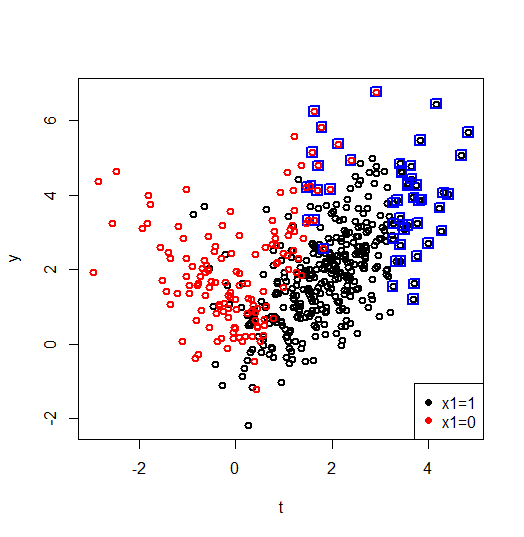}
\includegraphics[scale=0.6]{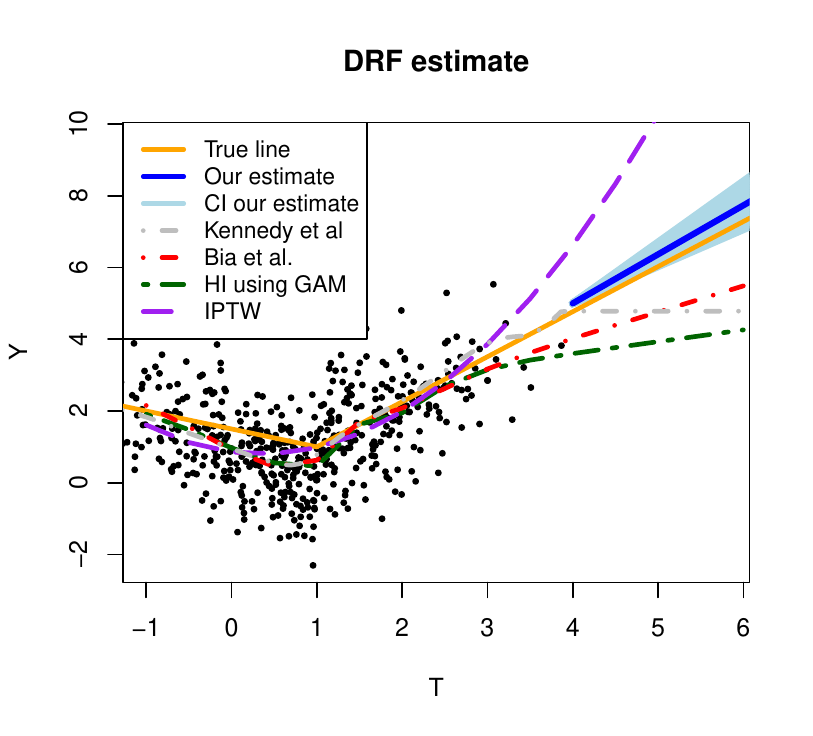}
\caption{\textbf{Left}: Data generated based on the simulations outlined in Section~\ref{Simulations_simple_example} with $n=500$. Points falling within the set $S$ are identified by a blue square.
\textbf{Right}: Estimation of $\mu(t)$ using various methods: Orange represents the true $\mu(t)$, blue depicts our estimate employing the method from Section~\ref{section_algorithm} with $95\%$ confidence intervals, grey illustrates the doubly robust estimation method introduced by \cite{Kennedy_2016, kennedy2019}, red showcases the additive spline estimator described in \cite{bia2014stata}, dark green demonstrates the approach proposed by \cite{Hirano} utilizing a GAM outcome model (further details in \cite{causaldrf}), and purple describes the inverse probability of treatment weighting estimator \citep{vanderwal2011ipw}.}
\label{Plot1}
\end{figure}

The subsequent illustrative example outlines our methodology and the main ideas. Consider a single confounder $X = X_1 \sim \text{Bernoulli}(0.75)$ (where $X_1=1$ denotes men and $X_1=0$ denotes women, for instance). Define $T = X_1 + \varepsilon_T$, where $\varepsilon_T \sim \mathcal{N}(0,1)$ (indicating that $T$ generally tends to be larger for men than for women). Let
  \begin{equation*}
     Y = \begin{cases}
      & T+\varepsilon,\,\,\,\,\,\,\,\,\,\,\,\,\,\,\,\,\,\,\,\,\,\,\,\,\,\,\, \text{when }X_1=1, \,\,T>1,\\ 
     & 2T+\varepsilon,\,\,\,\,\,\,\,\,\,\,\,\,\,\,\,\,\,\,\,\,\,\,\,\, \text{when }X_1=0,\,\, T>1,\\ 
     & 2-T  +\varepsilon,\,\,\,\,\,\,\,\,\,\,\,\,\,\, \text{when }T\leq 1,
    \end{cases}
    \end{equation*}
where $\varepsilon \sim \mathcal{N}(0,1)$. 

Simple computation gives us $\mu(t) = 0.75t + (1-0.75)2t = 1.25t$ for any $t>1$, while $\mu(t) = -t + 2$ for $t\leq 1$. Consequently, our primary interest lies in estimating the slope $$\omega=\mu(t+1) - \mu(t) = 1.25\,\,\,\,\, \text{ for }t>1.$$ 
We generate data as specified with a sample size of $n=500$. Setting the threshold at $q = 0.9$, we employ the methodology outlined in Section~\ref{section_algorithm} to estimate $\omega$. This process is repeated 100 times, yielding a mean and $0.95$ quantile of $$ \hat{\omega} = 1.26 \pm 0.39.$$ Additionally, we employ the bootstrap technique to calculate confidence intervals, and we obtain a confidence interval of the form $\omega\in(0.72, 1.87)$ on average. We see from other simulations that these confidence intervals are slightly conservative for $n\leq 1000$, but work well for larger sample sizes. 

In Figure~\ref{Plot1}, we present one generated dataset with a sample size of $n=500$, showcasing various methodologies from existing literature and their extrapolation efficacy. Notably, classical techniques often exhibit a tendency to underestimate $\mu(t)$ for $t$ large, primarily due to the fact that only the 'men' category ($X=1$) received $T>2$. 

We conclude with an important remark regarding the sample size: a substantial amount of valuable information is lost when we discard $90\%$ of the data by focusing solely on the data in the set $S$ (data above the threshold $\tau(\textbf{x})$). This is the primary reason behind the considerably large confidence intervals and the heightened variability in our estimates. We encounter the inevitable bias-variance trade-off; the inclusion of more data introduces a potential bias, given that the behavior of $\mu(t)$ differs in the body and in the tail.

\subsection{Simulations}

We provide a comprehensive discussion of all simulations in detail in Appendix~\ref{Appendix_Simulations}. In our study, we devised several simulation setups to model diverse scenarios and explore them thoroughly. Specifically, we focused on five key scenarios:

\begin{enumerate}
    \item  Investigating how our method scales with respect to the dimension of the confounders $d=dim(\textbf{X})$.
    \item Comparing our method with classical methods from the literature, 
    \item Expanding upon the simple example introduced in Section~\ref{Simulations_simple_example}, wherein we evaluated performance across various dependence structures (employing different copulas), sample sizes, and a spectrum of causal effects.
    \item  Examining the presence of a hidden confounder affecting both $T$ and $Y$.
    \item  Focusing on variations in the function $\mu(t)$.
\end{enumerate}

In this section, we present two key findings from our simulation study. Table~\ref{table_simulations_dimension_d} illustrates how our methodology scales with varying dimensions of the confounders  $d=dim(X)$. Additionally, Table~\ref{table_classical_methods_result} depicts the comparison of our method with other classical methods from the literature, where we estimated $\mu(\tilde{t})$ for $\tilde{t} = \max_{i=1, \dots, n}t_i + 10$.

\begin{table}[H]
\centering
\begin{tabular}{|c|c|c|c|}
\hline
True $\omega=-1$& Gaussian $\varepsilon_T$         & Exponential $\varepsilon_T$    & Pareto $\varepsilon_T$         \\ \hline
$d=5$       & $\hat{\omega} =-1.0 \pm 0.03$& $\hat{\omega} =-1.0 \pm 0.01$& $\hat{\omega} =-1.0 \pm 0.001$\\ \hline
$d=25$      & $\hat{\omega} =-0.96 \pm 0.08$   & $\hat{\omega} =-0.99 \pm 0.01$ & $\hat{\omega} =-1.0 \pm 0.001$\\ \hline
$d=50$      & $\hat{\omega} = -0.75 \pm 0.28 $ & $\hat{\omega} =-0.97 \pm 0.09$ & $\hat{\omega} =-0.99 \pm 0.01$ \\ \hline
$d=200$     & $\hat{\omega} = -0.44\pm 0.37 $  & $\hat{\omega} =-0.53 \pm 0.42$ & $\hat{\omega} =-0.91 \pm 0.68$\\ \hline
\end{tabular}
\caption{Estimates of $\omega=-1$ with varying dimension of the confounders $d=dim(X)$ and with different distributions of the noise of $T$. The sample size  is $n=5000$. The full simulations setup can be found in Appendix~\ref{Section_simulations_high_dimension_confounder}.}
\label{table_simulations_dimension_d}
\end{table}

\begin{table}[H]
\centering
\begin{tabular}{|l|c|c|c|c|c|}
\hline
 & \multicolumn{1}{l|}{Our method} & \multicolumn{1}{l|}{Bia et al.} & \multicolumn{1}{c|}{Kennedy et al.} & \multicolumn{1}{l|}{HI with GAM} & \multicolumn{1}{l|}{IPTW} \\ \hline
$d=2$  & \textbf{0.18}& 0.68&   0.64& 0.42& 3.89\\ \hline
$d=10$ & \textbf{0.48} & 0.81&   0.65& 0.67 & 5.69 \\ \hline
$d=30$ & \textbf{0.79}& 0.92& could not handle& 0.92& 4.70\\ \hline
\end{tabular}

\caption{Comparing the extrapolation performance of various models using the average Absolute Relative Error (ARE) across 100 simulations with varying dimension of the confounders $d=dim(\textbf{X})$. The interpretation of the values is as follows: if the true value of $\mu(\tilde{t})$ is 1, an ARE of 0.17 indicates an approximate typical error of $|\hat{\mu}(\tilde{t}) - \mu(\tilde{t})| \approx 0.17$.}
\label{table_classical_methods_result}
\end{table}

Table~\ref{table_simulations_dimension_d} suggests that the results are reasonably accurate as long as $d\leq 50$. As discussed in the Appendix~\ref{Appendix_Simulations}, the reason for the bias observed in larger dimensions $d$ is that Assumption~\ref{assumption_max_domain} and (\ref{eq1}) are only asymptotic results, and with higher dimensions $d$, we require more data for the asymptotic theory for $T\mid \textbf{X}$ to take effect. It is well known that the convergence rate of the maxima of the Gaussian random sample to an extreme value distribution is very slow, whereas it is faster with Exponential or Pareto distribution \citep{Davis_converngence}. 

Analysis of Table~\ref{table_classical_methods_result} reveals that our method achieves the smallest extrapolation error. Hirano and Imbens \citep{Hirano} method utilizing a GAM outcome model showed surprisingly reasonable performance, while the IPTW method \citep{vanderwal2011ipw} exhibits poor performance due to the quadratic nature of the extrapolating curve.  Please note that our method assumes linearity in the tail. If this assumption isn't met, our method may produce inferior results. In such instances, alternative regression techniques can be utilized in step 2 of our algorithm instead of linear regression. For example, neural networks, as demonstrated in \cite{XinweiShen2023engression}, can be employed to address nonlinear behavior and enhance performance.

\section{Application: River discharge dataset}
\label{section_application2}

\begin{wrapfigure}{r}{5cm}
\includegraphics[scale=0.2]{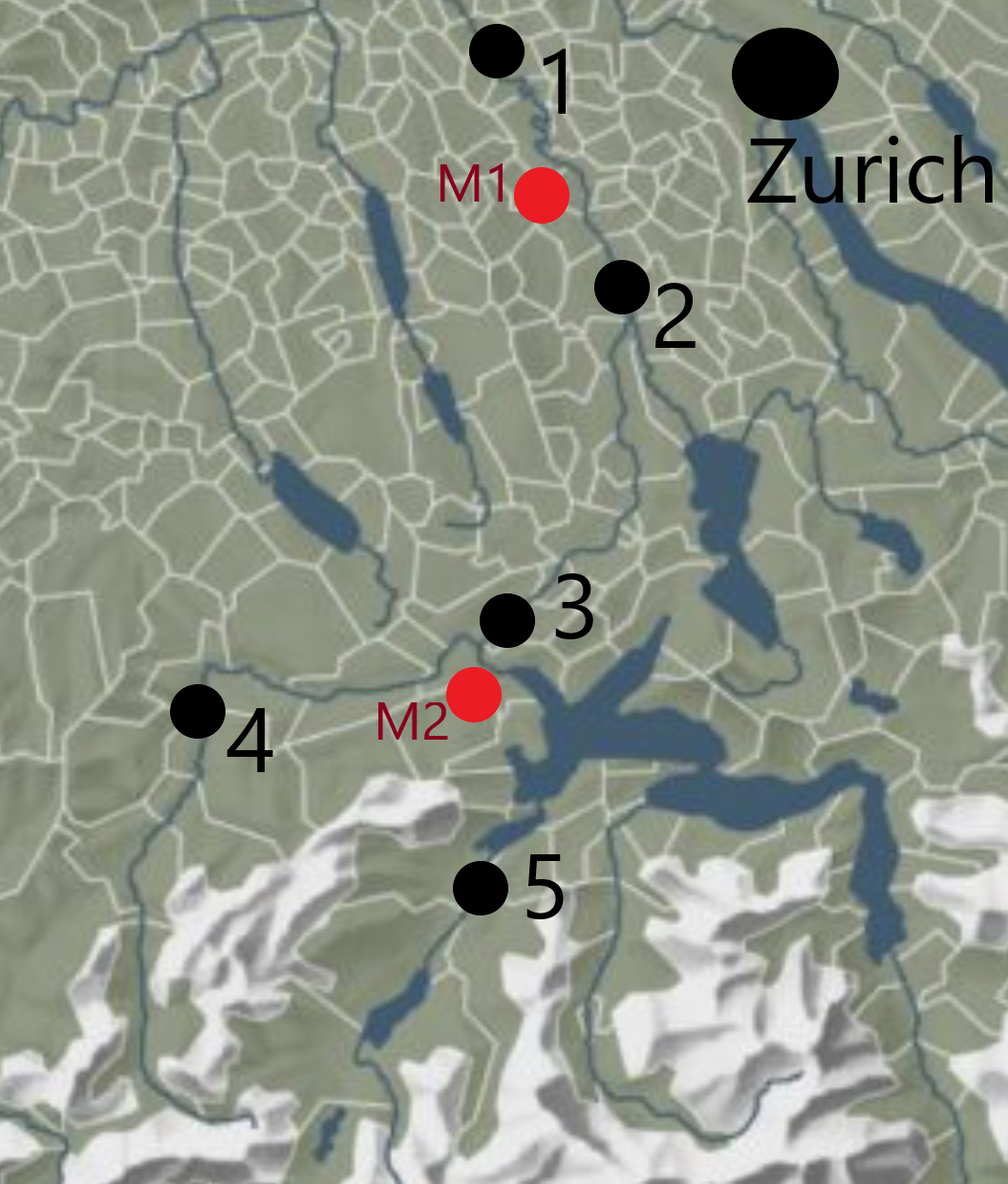}
\caption{Map of meteo-stations (red) and five river stations (black). Note that the river flow is from south to north (with springs in the mountains). }
\label{map}
\end{wrapfigure} 
Understanding the causal relationship between extreme precipitation and river discharge is crucial for effective water resource management. In this study, we examine how extreme precipitation events impact river discharge. By utilizing a comprehensive dataset spanning various hydro-logical conditions, our research seeks to provide insights into the critical nexus between extreme precipitation dynamics and extreme river discharge events. The data were collected by the Swiss Federal Office for the Environment (\url{hydrodaten.admin.ch}), but were provided by the authors of  \cite{Pasche, engelke2021sparse}, with some useful preliminary insights. We used precipitation data and other relevant measured variables from meteorological stations provided by Swiss Federal Office of Meteorology and Climatology, MeteoSwiss (\url{gate.meteoswiss.ch/idaweb}).

We exclusively examine the discharge levels of the River Reuss, situated near Zurich in Switzerland (Figure~\ref{map}). We selected this river due to the availability of excellent measurements of its discharge levels, complemented by well-documented weather conditions from nearby meteorological stations and diverse landscape. Our measurements include average daily discharges between January 1930 and December 2014 and daily precipitation in the nearby meteo-stations.  To reduce any seasonal effects due to unobserved confounders, we only consider data during June, July and August, as the more extreme observations happen during this period when mountain rivers are less likely to be frozen. 

We center our attention on addressing two distinct research questions: one characterized by a straightforward scenario where the ground truth is known, and another presenting a more intriguing challenge. 

\subsection{Known ground truth}

We demonstrate our methodology using a straightforward example where the ground truth is known. Consider a pair of river stations, such as stations 2 and 1. Let $T$ represent the water discharge at station 2, $Y$ represent the water discharge at station 1, and $\textbf{X}$ denote measurements taken at a nearby meteorological station (including precipitation, humidity, etc. The full list of confounders can be found in Appendix~\ref{Appendix_application2}). Our objective is to investigate the impact of extreme discharge levels at station 2 on the water discharge observed at station 1. In mathematical terms, we seek to ascertain $\mu(t)$ or $\mu_\textbf{x}(t)$ for large values of $t$. In this context, the ground truth is the following: $$\mu(t_1+t_2) - \mu(t_2) = \mu_\textbf{x}(t_1+t_2) - \mu_\textbf{x}(t_2) = t_1-t_2$$ for all $t_1, t_2 \geq 0$ and all $\textbf{x} \in \mathcal{X}$. This can also be explained in words as follows: if we pour $t_1$ liters of water into the river at station 2 (in causal terminology, we interpret this as an intervention $do(T = T+t_1)$), we expect the water discharge at station 1 ($Y$) to increase by exactly $t_1$. Hence, $\omega = \omega_\textbf{x}=1$. As we will see below, our methodology consistently yields this expected outcome.

We follow the methodology introduced in Section~\ref{section_algorithm} with $q=0.95$.  Detailed steps, diagnostics and preliminary data analysis (just for a pair $2\to 1$) can be found in Appendix~\ref{Appendix_application2}. The resulting estimates can be found in Table~\ref{table_pairs}. The results are very similar with a different choices of $q$ (changing $\hat{\omega}$ by not more than by $0.1$). We observe that our results align very well with the ground truth ($\omega = 1$).  However, there is a slight bias evident in the relationships between the pairs $5\to 3$ and $4\to 3$: this can be attributed to distinct geographical features. Notably, a lake Vierwaldstättersee lies between the pair $5$ and $3$, which diminishes the influence of $5$ on $3$. Additionally, a 3238m Titlis mountain is situated between pair $4$ and $3$, amplifying the effect of $4$ on $3$ due to the melting glacier ice, acting as an unmeasured confounding factor. Our methodology relies on Assumptions~\ref{assumption_max_domain}, \ref{assumption_unconfoundness_tail}, \ref{assumption_positivity_tail}, and \ref{assumption_linearity_of_tail_1}, along with some continuity assumptions and the SUTVA assumption discussed in Section~\ref{section_problem_statement}. Assumptions~\ref{assumption_max_domain} and \ref{assumption_positivity_tail} are minor and are used frequently when dealing with these types of data. Assumption~\ref{assumption_unconfoundness_tail} is a common and challenging aspect of every causal inference methodology. While our assumption is weaker than the classical unconfoundness assumption (requiring no hidden confounder in the tail), complete rejection of the possibility of its violation is unattainable. However, we believe that the meteo-station between a pair of stations can capture the most significant confounders (with the exception when the lake or mountains are present in between the river-stations). Finally, Assumption~\ref{assumption_linearity_of_tail_1} is a strong assumption that allows us to extrapolate observed values into the extremal region. However, this assumption (or at least some similar model assumptions) are necessary. In this case, linear assumption is valid, since the underlying ground truth is known. 

\begin{table}[h]
\begin{tabular}{|c|c|c|c|c|}
\hline
Truth: ${\omega} = 1$ & Stations $2\to 1$& Stations  $3\to 2$& Stations  $4\to 3$& Stations $5\to 3$\\ \hline
$\hat{\omega}$        & $1.03\pm 0.05$& $1.17 \pm 0.24$& $1.21\pm 0.19$& $0.78\pm 0.41$\\ \hline
\end{tabular}
\caption{Estimates $\hat{\omega}$ between each pairs of the stations. }
\label{table_pairs}
\end{table}

\subsection{Effect of precipitation on river discharge}

We employ our methodology to address a more complex inquiry where the ground truth is not known. Let's consider water discharge at station 3 ($Y$) and let $T$ denote the precipitation measured in meteo-station M2. Our focus lies in understanding the impact of extreme precipitation events ($T$) on the water discharge ($Y$). As mentioned in the introduction, on 6.6.2002, we recorded a historical maximum precipitation level of $T_{\text{max}} = 111 \frac{mm}{m^2}$, coinciding with the scenario when the river nearly breached its banks. Our inquiry centers on the question: how would the river discharge $Y$ alter if $T$ were to reach $120 \frac{mm}{m^2}$? In mathematical terminology, we are interested in estimating $\mu(120) - \mu(111)$ or possibly $\mu_{\textbf{x}^\star}(120) - \mu_{\textbf{x}^\star}(111)$ where ${\textbf{x}^\star}$ are other covariates corresponding to that event. Addressing this question is challenging as we lack data within this extreme regime, necessitating reliance on extrapolation. This task is especially challenging, since we anticipate that the effect of precipitation on river discharge may vary between the body of the distribution and its tail, since the ground absorbs a significant portion of the rainfall during a light rain. 

We follow the methodology introduced in Section~\ref{section_algorithm}. 
A straightforward approach would be to define $T$ as precipitation and $Y$ the water discharge on the same day, while choosing appropriate confounders $\textbf{X}$ from some measurements at M2. Then, we can use classical approach for estimating $\mu(t)$ in the body and our approach to estimate it in the tail. However, some problematic issues arise in this application:
\begin{itemize}
    \item \textbf{Time issue:} $T_{monday}\to Y_{monday}$ but also  $T_{monday}\to Y_{tuesday}$ since it takes time for the rain water to reach the river and rain tends to be more frequent around midnight. In fact, correlation (and extreme correlation coefficient as well, see Figure~\ref{ccf}) is much higher for a pair ($T_{monday}, Y_{tuesday}$) than for  ($T_{monday}, Y_{monday}$). The extreme storm on 6.6.2002 corresponded to extremely high river discharge on 7.6.2002 (where $Y$ was about five times larger than on 6.6.2002).  Hence, our interest lies in the effect $T_{monday}\to Y_{tuesday}$ (that is, we consider $t_i$ as precipitation on day $i$ while $y_i$ is the discharge on day $i+1$). Additionally, the presence of time introduces an auto-correlation issue. This can be handled by taking for example weekly maxima or discarding consecutive observations within a certain time frame to reduce the auto-correlation effect. Alternatively, applying techniques like time series decomposition, differencing, or using autoregressive models can also mitigate the issue of auto-correlation in the data analysis process. We leave the data unchanged since the temporal dependence is primarily local, spanning only a few days, and does not introduce a substantial bias.
    \item \textbf{Variable selection issue:} choosing appropriate confounders $\textbf{X}$ that act as confounders of $Y$ and $T$. It is not clear which variables can be safely considered as confounders: if a variable $X$ lie on a path $T\to X\to Y$, adjusting for $X$ would lead to so-called path-cancelling causal effect \citep{direct_indirect_effects}.  Here, we are interested in so-called total causal effect, so we need to be cautious of which covariates to adjust for. However, not adjusting for a common cause leads to a bias. Moreover, there is often a feedback loop: $X_i\leftrightarrow Precipitation$ for $X_i$ for example humidity or temperature. However, some of the variables can be safely considered as common causes: for example temperature on Sunday (the day before measuring precipitation). There is a huge amount of literature for such a variable selection, and we do not aim to comment on this research area: we only provide a full list of chosen confounders in Appendix~\ref{Appendix_application2}. 
\end{itemize}

We estimate two values: $\hat{\omega}$ which is the tail quantity defined as difference between $\mu(t+1) - \mu(t)$ for large $t$: in our case, how would $Y$ change if it was raining by $1\frac{mm}{m^2}$ more on 6.6.2002? Next, we also estimate  $\hat{\beta} = \mu(t+1) - \mu(t)$ corresponding to the body of the distribution (see Appendix~\ref{application_river_computation_of_beta} for details on its computation). The resulting estimates can be found in Table~\ref{table_stations} and visualisation of the $\mu(t)$ can be found in Figure~\ref{estimation_in_the_body_river_data}. We observe that the effect of $T$ on $Y$ is larger in the extreme region than in the body of the distribution by a factor of $\frac{3.04}{2.4}\approx 1.25$. 

\begin{table}[]
\begin{tabular}{|c|c|c|c|c|c|}
\hline
Truth unknown& Station 1& Station 2& Station 3& Station 4&Station 5\\ \hline
$\hat{\beta}$& $2.4\pm 0.1$& $2.28\pm 0.1$& $1.44\pm 0.02$& $0.89\pm 0.02$&$0.38\pm 0.01$\\ \hline 
 $\hat{\omega}$        & $3.04\pm 0.95$& $2.61\pm 0.67$& $1.62\pm 0.35$& $0.99\pm 0.32$&$0.36\pm 0.13$*\\ \hline
\end{tabular}
\caption{Estimates $\hat{\beta}$ and $\hat{\omega}$ represent the estimation of the effect of $T$ on $Y$ in the body and in the tail, respectively.  $\hat{\beta}$ is computed using standard regression while  $\hat{\omega}$ is the tail counterpart computed using steps  introduced in Section~\ref{section_algorithm}. *Note that Station 5 is in different altitude and relatively far from meteo-station M2 with a lake in between: hence, there is a bias due to data collection problems.  }
\label{table_stations}
\end{table}

As for the answer to our question 'how would the river discharge $Y$ alter in station 3 if $T$ were to reach $120 \frac{mm}{m^2}$ on 6.6.2002', our results suggest that the water discharge would be larger by about $9\times 1.62 = 14.5\frac{m^3}{s}$ (note that median of $Y$ is $11.2$ and $95\%$ quantile of $Y$ is $51.2$). Would this result in the river overflowing its banks? We cannot definitively say, as we lack the necessary data regarding the volume and contours of the river banks. Moreover, $Y$ represent the daily average of the water discharge while in order to answer this question, daily maximum is a better suited variable for answering this question. Nonetheless, this advances us towards a more accurate understanding of the effects and impacts of extreme precipitation events and potentially enhancing statistical inference for hydroelectric power stations located along this river.

\begin{figure}[t]
\centering
\includegraphics[scale=0.6]{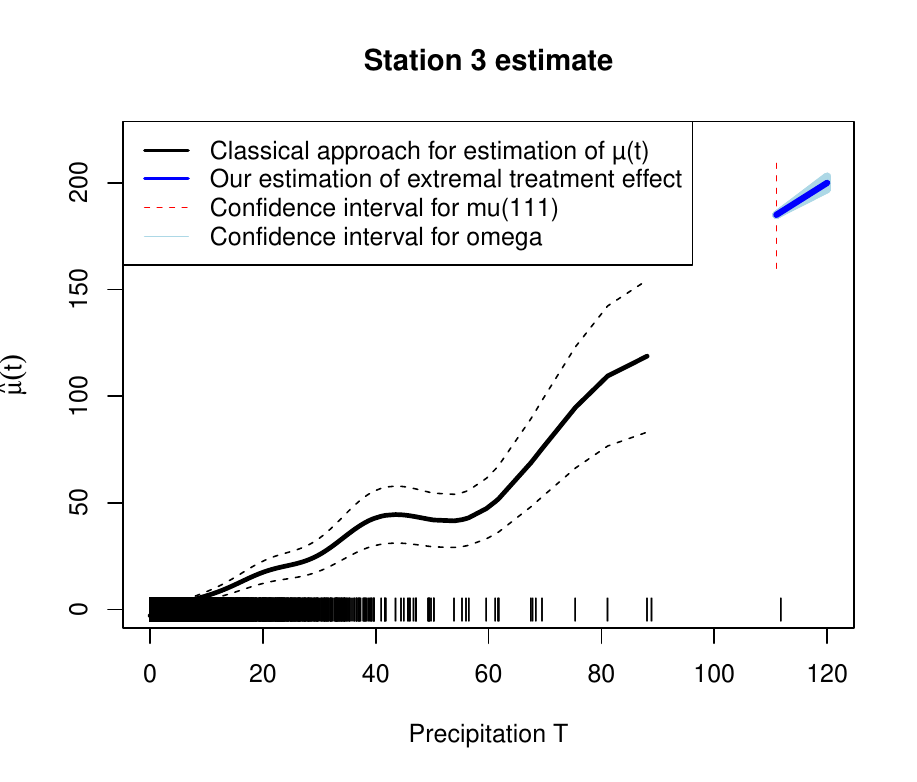}
\caption{The estimation of $\mu(t)$ using the doubly robust estimator introduced in \cite{Kennedy_2016, kennedy2019} cut at the second largest observation, together with estimation of $\mu(111), \omega$ and their $95\%$ confidence intervals.   }
\label{estimation_in_the_body_river_data}
\end{figure}

\section{Conclusion and future work}

Analyzing the impact of extreme levels of a treatment variable (exposure) is essential for comprehending its effects on diverse systems and populations. In this paper, we introduced a novel framework aimed at estimating the causal effect of extreme treatment values. Leveraging insights from extreme value theory, we enhanced the estimation of the extreme treatment effect. Our framework can handle a substantial number of confounders. Nonetheless, our methodology relies on extrapolation, presenting inherent challenges even in the absence of confounding variables, where the bias stemming from a model misspecification is impossible to quantify. Our framework holds promise for initial assessments of the impact of extreme environmental events, such as the effects of severe storms or droughts on economic damages. Future work may explore the application of our extreme value theory approach to address time-varying effects, a prevalent issue in environmental research.

\section*{Conflict of interest and data availability}
The code is available in the \href{https://github.com/jurobodik/Extreme_treatment_effect.git}{online repository} or on request from the author, alongside with the data regarding Appendix~\ref{section_application}. While the data related to the application discussed in Section~\ref{section_application2} are not publicly available, they can be accessed through \url{hydrodaten.admin.ch} and \url{gate.meteoswiss.ch/idaweb} after registration or by requesting the used data from the authors of \cite{Pasche}. 

The authors declare that they have no known competing financial interests or personal relationships that could have appeared to influence the work reported in this paper.

\section*{Acknowledgements}
This study was supported by the Swiss National Science Foundation. The author would like to thank Prof. Valérie Chavez for her mentorship and Prof. Mats Stensrud for his insightful discussions.

\appendix
\newpage
\section{Appendix: Application 2 - concrete compressive strength}
\label{section_application}

\subsection{Main analysis}

In this section, we delve into a dataset \citep{Concrete_data} focused on concrete compressive strength  \footnote{We express our gratitude to Adam Kovac for suggesting and discussing this dataset.}.  

Concrete serves as a fundamental material in civil engineering, and understanding its compressive strength (denoted as $Y$ and measured in MPa) is paramount for ensuring structural integrity \citep{Concrete}. Concrete comprises ingredients such as cement ($X_1$), fly ash ($X_2$), water ($X_3$), superplasticizer ($X_4$) and blast furnace slag ($T$) (amongst some other additions). The units of $T, X_1, X_2, X_3, X_4$ are in kilograms in a $m^3$ mixture. The concrete compressive strength ($Y$) exhibits a highly nonlinear relationship with these ingredients and the elapsed time. Our focus is on exploring the effect of blast furnace slag ($T$) on compressive strength ($Y$). It is well-established that increasing the quantity of $T$ can enhance $Y$, yet an excessive amount of $T$ may lead to a decrease in $Y$. 

In our dataset $X_1, X_2, X_3, X_4$ may affect the decision of how much $T$ was used (engineers often decide about the quantity of $T$ based on the looks of the mixture of other ingredients). Our dataset contains $n=1030$ instances of observational data $\{\textbf{x}_i, t_i, y_i\}_{i=1}^n$ where $\textbf{x}_i = (x_{1,i}, \dots, x_{4,i})^\top$.  The range $(\min_iy_i, \max_iy_i) = (2.3, 82.5)$ and $(\min_it_i, \max_it_i) = (0, 359)$.

Suppose we fit a linear model $\mathbb{E}Y = \beta_0 + \beta_TT+\beta_1X_1+ \beta_2X_2+ \beta_3X_3+ \beta_4X_4$; then, a least square estimation of the coefficient ${\beta}_T$ leads to $\hat{\beta}_T = 0.08 \pm 0.006$. This can be (wrongly) interpreted as 'adding one additional kg of $T$ in $m^3$ mixture increases $Y$ by $0.08MPa$'. We expect different behaviour for small and large values of $T$, we expect strong (nonlinear) interactions between the covariates and, more importantly, this result is derived from the body of the distribution, while we are interested in values of $T$ above the observed ones. 

Our objective is to quantify the effect of an extreme amount of blast furnace slag $T$ on $Y$. Specifically, we answer the following questions: 
\begin{enumerate}
    \item Given a concrete mixed with $T=359$ and $\textbf{X}=\textbf{x}$ for some specific value of $\textbf{x}$, if we intervene and change $T$ to $T=400$, what effect on concrete compressive strength can we expect? Using the potential outcome notation, the quantity of interest is $\mu_{\textbf{x}}(400) - \mu_{\textbf{x}}(359)$. Note that $max_{i=1, \dots, n}t_i=359$ (we don't observe the blast furnace slag larger than $359$, and there is no observation in the interval $(220, 359)$) and hence, we have zero data in such an extreme region. We aim to answer this question for a choice $\textbf{x} = \textbf{x}^\star$ where $x_1^\star=239, x_2^\star= 0, x_3^\star=185, x_4^\star=0$ (the observation corresponding to $T_i=359$).    
    \item How would an extreme increase in $T$ change $Y$ for an 'average' concrete (On a population level, i.e. integrating over the covariates)? Using the potential outcome notation, the quantity of interest is  $ \mu(400) - \mu(359)$.
\end{enumerate}

We follow the methodology introduced in Section~\ref{section_algorithm} with $q=0.9$.  Detailed steps, diagnostics and preliminary data analysis can be found in Section~\ref{Appendix_app_2}. The resulting estimates are as follows: 
\begin{align*}
     \hat{\mu}_{\textbf{x}^\star}(400) -  \hat{\mu}_{\textbf{x}^\star}(359)   &= -4.1\pm 3.0, &  \hat{\mu}(400) -  \hat{\mu}(359)  &= -4.5\pm 2.6, \\
     \hat{\omega}_{\textbf{x}^\star}  &= -0.1\pm 0.07, & \hat{\omega} &= -0.11\pm 0.06. 
\end{align*}

The results are similar with a different choices of $q$ (see Table~\ref{Table_application1}). In summary, the results suggest that for a mixture of concrete with covariates $\textbf{X} = \textbf{x}^\star$ and $T=359$, intervening on $T$ and changing it to $T=400$ would decrease the concrete compressive strength by about $4.1MPa$. On the population level, increasing $T$ from $359$ to $400$ would lead to decrease in the concrete compressive strength by about $4.5MPa$. The $95\%$ confidence intervals suggest that this estimate can be inaccurate by about $3MPa$; however, one must be cautious about the interpretation of the confidence intervals, since they are in general unreliable when extrapolating, see Remark~\ref{remark_error}.

Figure~\ref{Kennedy+me} graphically shows the estimation of $\hat{\mu}(t)$ in the body using the method introduced in \cite{Kennedy_2016, kennedy2019}, as well as our estimation of $\hat{\mu}(t)$ for extreme values.

In the Appendix~\ref{Appendix_app_4}, we discuss the assumptions made regarding our results. In brevity, our methodology uses Assumptions~\ref{assumption_max_domain}, \ref{assumption_unconfoundness_tail}, \ref{assumption_positivity_tail}, and  \ref{assumption_linearity_of_tail_1} (amongst some continuity assumptions and SUTVA assumption discussed in Section~\ref{section_problem_statement}).  While we argue that assumptions~\ref{assumption_max_domain} and \ref{assumption_positivity_tail} are minor, we can not disregard the possibility of a hidden confounder (assumption~\ref{assumption_unconfoundness_tail}). The validity of assumption~\ref{assumption_unconfoundness_tail} has to be further argued by an expert knowledge. Finally, the strongest assumption is assumption~\ref{assumption_linearity_of_tail_1}, as its violation leads to the most significant bias. However, this assumption (or a similar assumption using different model) is necessary when extrapolating and it is hypothetically testable by measuring values with $T\approx 400$. Appendix~\ref{Appendix_app_3} also discuss the differences for the range of choices of $q$.

\begin{figure}[]
\centering
\includegraphics[scale=0.5]{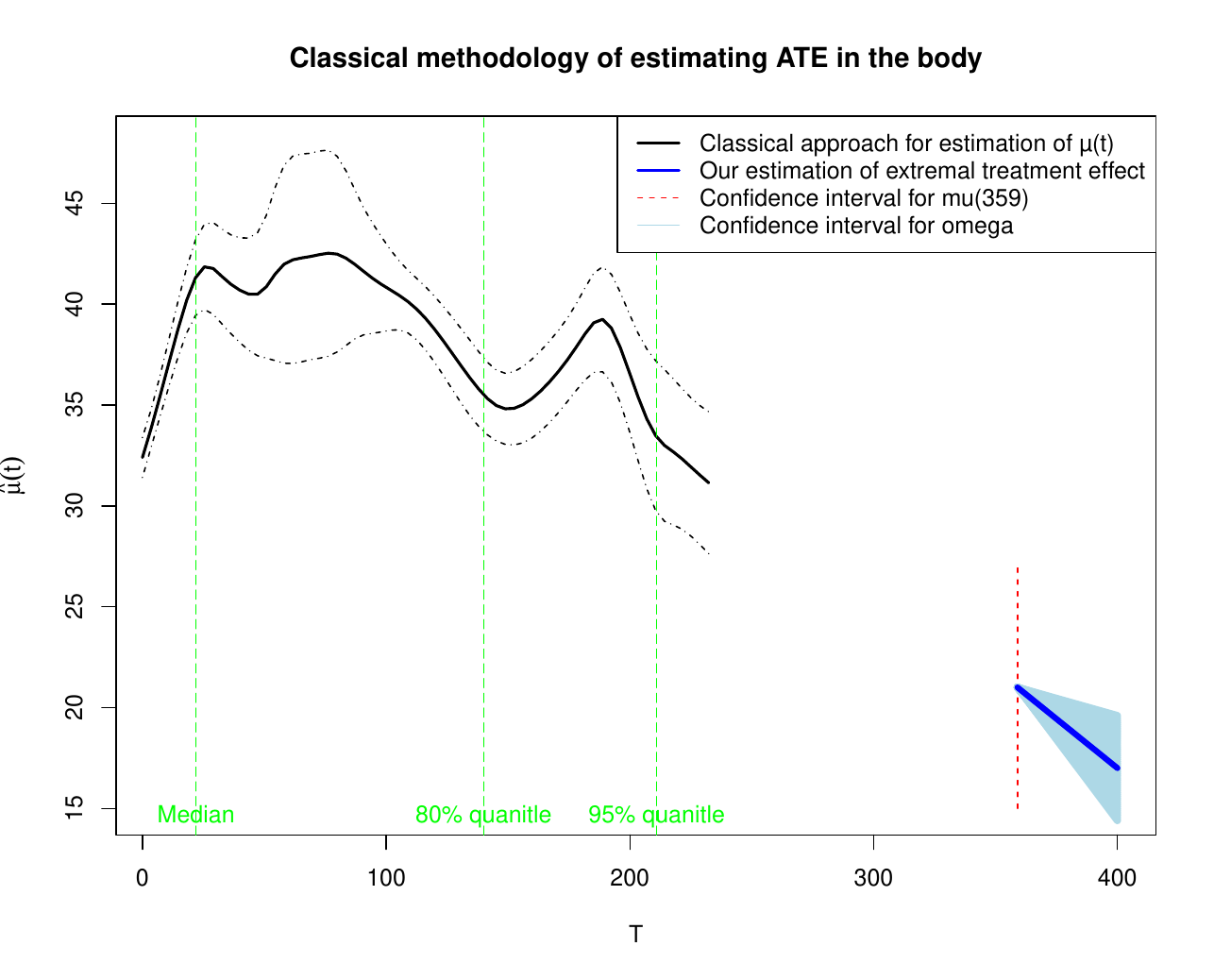}
\caption{\textbf{Black:} The estimation of $\mu(t)$ using the doubly robust estimator introduced in \cite{Kennedy_2016, kennedy2019}. \textbf{Green}: Quantiles of $T$. \textbf{Blue:} Our estimation of $\mu(t)$ for values $t=359, 400$ for $q=0.9$, together with the $95\%$ confidence intervals for the slope. \textbf{Red:} $95\%$ confidence interval for $\mu(359)$.}
\label{Kennedy+me}
\end{figure}

\subsection{Detailed computations of the estimates}
\label{Appendix_app_2}
Some data visualisation can be found in Figures~\ref{Plot22} and \ref{Plot23}.  In the following, we provide detailed descriptions of the steps undertaken in the application for a specific choice  $q=0.9$. First, we estimate $\tau(\textbf{x})$ using a classical quantile regression \citep{koenker_2005}. We observe that all covariates are highly significant and the diagnostic plots do not show any significant problems (except the fact that for many observations is $T_i=0$): we illustrate the estimation on Figure~\ref{Plot33}, where points above the $90\%$ threshold (points in the set $S$) are marked. 

In the next step, we routinely estimate $\theta(\textbf{x})$ \footnote{Using $\texttt{evgam}$ function in $\texttt{evgam}$ package \citep{EVGAM} using the following code: \texttt{evgam( list($T_{e}\sim s(X1_{e})+s(X2_{e})+s(X3_{e})+s(X4_{e}), \sim 1), data =data.frame(T_{e}, X_{e})$, family="gpd"  )} where $T_{e}$ are datapoints in $S$ (above the estimated $90\%$ threshold). } . More precisely, we assume fixed $\xi(\textbf{x}) = \xi\in\mathbb{R}$ and only estimate $\sigma(\textbf{x})$ as a smooth function of the covariates; Plot~\ref{Plot44} shows the estimated values of $\sigma(\textbf{x})$ on a log-scale. 

Finally, following the expression  $\mathbb{E}[Y\mid T=t, \textbf{X}=\textbf{x}] = \alpha[\hat{\theta}(\textbf{x})] + \beta[\hat{\theta}(\textbf{x})]t$, we estimate $\alpha, \beta$ from the data-points in $S$ using $\texttt{gam}$ function \citep{Wood}. Under assumptions~\ref{assumption_max_domain}, \ref{assumption_unconfoundness_tail}, \ref{assumption_positivity_tail}, and  \ref{assumption_linearity_of_tail_1} , we obtain $\hat{\mu}_{\textbf{x}^\star}(400) = \hat{\alpha}[\hat{\theta}(\textbf{x})] + \hat{\beta}[\hat{\theta}(\textbf{x})]400$, and in effect, we return $$\hat{\omega}_{\textbf{x}^\star} = \hat{\mu}_{\textbf{x}^\star}(400) - \hat{\mu}_{\textbf{x}^\star}(359)= \hat{\beta}[\hat{\theta}(\textbf{x}^\star)](400-359) = -4.1.$$ 
Regarding the second question (population level), we simply take the average  $\hat{\mu}(t):=\frac{1}{n}\sum_{i=1}^n  \hat{\alpha}[\hat{\theta}(\textbf{x}_i)]+ \hat{\beta}[\hat{\theta}(\textbf{x}_i)]t$ and compute 
$$\hat{\omega} = \hat{\mu}(400) - \hat{\mu}(359)= \frac{1}{n}\sum_{i=1}^n \hat{\beta}[\hat{\theta}(\textbf{x}_i)](400-359) = -4.5.$$Regarding the confidence intervals, we resample the data using the code: $$\texttt{resampled\_data = sample\_n(data, size = length(y), replace = TRUE)}.$$Then, we follow the same steps as above and estimate the coefficients (from the resampled dataset). We repeat this procedure 500 times. Finally, we take the $95\%$ quantile out of all computed resampled coefficients. For example, $\hat{\omega}_{\textbf{x}^\star} = -6.1 \pm 3.2$ represents the fact that the $95\%$ quantile was $-6.1 + 3.2 = -2.9$, and hence, only $5\%$ of values were larger than $-2.9$.

\subsection{Discussion about the results regarding different threshold q}
\label{Appendix_app_3}

\begin{table}[]
\centering
\begin{tabular}{|c|c|c|} \hline 

$q=0.85$                 & $q=0.9$                  & $q=0.95$                 \\ \hline  
 $\hat{\omega}_{\textbf{x}^\star} = -6.1 \pm 3.2$& $\hat{\omega}_{\textbf{x}^\star} = -4.1\pm 3.0$&$\hat{\omega}_{\textbf{x}^\star} = - 2.8 \pm 2.8$\\ \hline 
$\hat{\omega} = -5.3 \pm 4.0$& $\hat{\omega} = -4.5\pm 2.6$& $\hat{\omega} = - 3.3 \pm 3.1$\\ \hline
\end{tabular}
\caption{Estimates of $\omega_{\textbf{x}^\star}:= \mu_{\textbf{x}^\star}(400) - \mu_{\textbf{x}^\star}(359)$  for different thresholds $q$, together with the corresponding $95\%$ confidence intervals.}
\label{Table_application1}
\end{table}
Table~\ref{Table_application1} shows the results for different choices of $q$. Even though they yield distinct estimates for $\hat{\omega}$, the confidence intervals overlap, and the values $\hat{\omega} \in (-6.4, -1.9)
$ are encompassed by all of them. This suggests some stability in the choice of $q$.

The selection of $q$ reflects the bias-variance tradeoff; as we increase $q$, our inference relies on values closer and closer to $T=400$ (datapoints with small and intermediate $T_i$ can bias our estimation since in this region, increasing $T_i$ can increase $Y$). However, increasing $q$ also means disregarding more and more datapoints, and our estimate will have less power and larger variance.

The challenge of choosing $q$ is a common problem in extreme value theory, and the rule of thumb is to select $q$ as large as possible while maintaining an adequate number of datapoints above the $1-q$ quantile to ensure reasonably good inference.

\subsection{Discussion about the assumptions}
\label{Appendix_app_4}

Our methodology relies on Assumptions~\ref{assumption_max_domain}, \ref{assumption_unconfoundness_tail}, \ref{assumption_positivity_tail}, and \ref{assumption_linearity_of_tail_1}, along with some continuity assumptions and the SUTVA assumption discussed in Section~\ref{section_problem_statement}. Below, we provide a detailed discussion of each assumption.
\begin{enumerate}
    \item Assumptions~\ref{assumption_max_domain} and \ref{assumption_positivity_tail} are considered minor. As mentioned in Section~\ref{section_problem_statement}, Assumption~\ref{assumption_max_domain} is satisfied for most common distributions, and similar model assumptions are imposed in almost all applications utilizing extreme value theory. Assumption~\ref{assumption_positivity_tail} appears to be satisfied, as there is no specific range of values in the support of $T$ that has zero probability of occurring.
    \item  Assumption~\ref{assumption_unconfoundness_tail} is a common and challenging aspect of every causal inference methodology. While our assumption is weaker than the classical unconfoundness assumption (requiring no hidden confounder in the tail), complete rejection of the possibility of its violation is unattainable. A potential hidden confounder could be the 'quality of ingredients.' If the quality is low, engineers might tend to use excessive amounts of $T$ in the mixture, potentially leading to spurious dependence between large $T$ and low $Y$. However, in this case, it seems plausible that this hidden dependence due to low ingredient quality does not introduce a substantial bias. An expert knowledge is required to ensure the validity of this assumption. 
    \item Assumption~\ref{assumption_linearity_of_tail_1} is a strong assumption that allows us to extrapolate observed values into the extremal region. However, this assumption (or at least some similar model assumptions) are necessary; estimating $\mu(400)$ from observed values is not feasible otherwise. In essence, Assumption~\ref{assumption_linearity_of_tail_1} asserts that the relationship between $T$ and $Y$ (given other confounders) is linear in the unobserved region below $T=400$. Since there is no other reason to believe that this relationship has any particular form, a linear assumption seems to be the most suitable choice. Although this assumption is strong, it is hypothetically possible to test by measuring values with $T\approx 400$.
\end{enumerate}

\begin{figure}[H]
\centering
\includegraphics[scale=0.5]{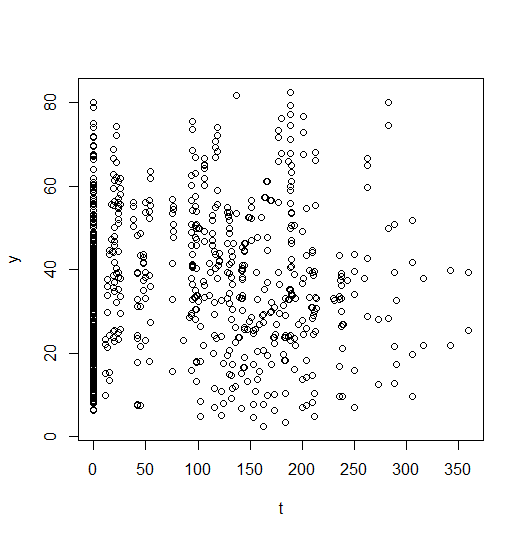}

\caption{The figure illustrate the dependence among $T$ and $Y$. Note that the correlation between $T$ and $Y$ is $0.14\pm 0.07$.  }
\label{Plot22}
\end{figure}

\begin{figure}[H]
\centering
\includegraphics[scale=0.5]{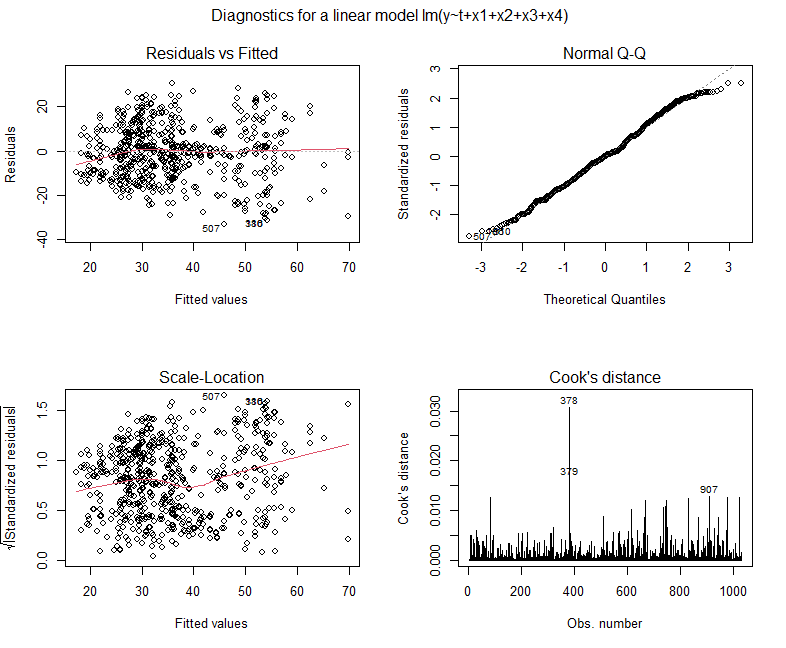}
\caption{Diagnostics of a linear model fitted into the original data.  }
\label{Plot23}
\end{figure}

\begin{figure}[H]
\centering
\includegraphics[scale=0.5]{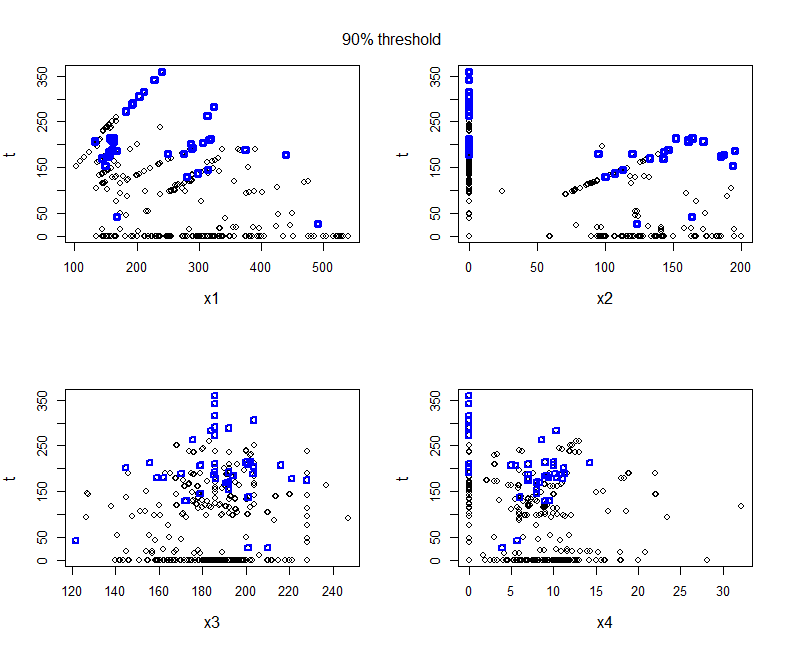}
\caption{ Visualisation of the estimation of $\tau(\textbf{x})$- estimated using classical quantile regression. Blue points characterise the observations above this threshold (points in the set $S$).   }
\label{Plot33}
\end{figure}

\begin{figure}[H]
\centering
\includegraphics[scale=0.5]{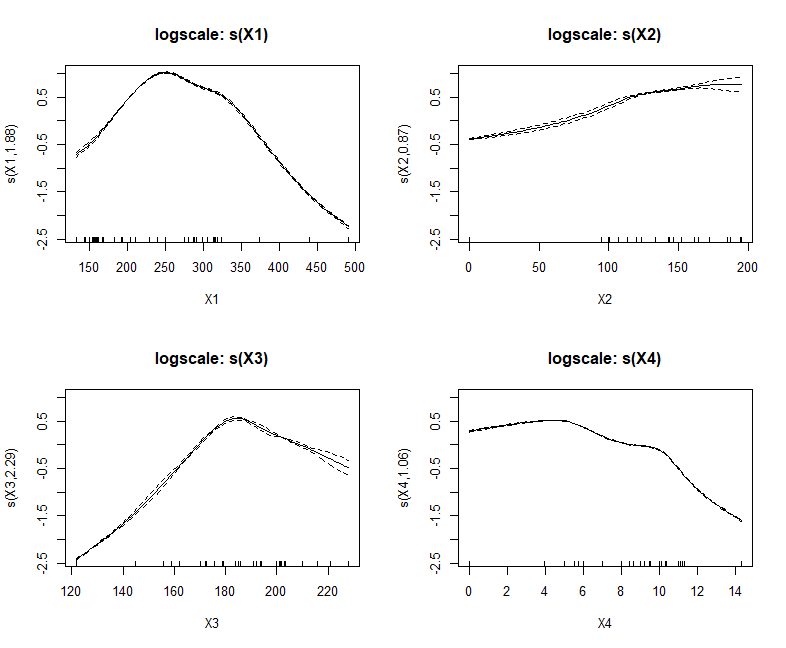}
\caption{ Estimation of scale in the tail-model.    }
\label{Plot44}
\end{figure}

\section{Appendix: Simulations}
\label{Appendix_Simulations}
In this section, we create various simulations setups to assess the performance of our methodology. 
\begin{wrapfigure}{r}{5.5cm}
\includegraphics[width=5.5cm]{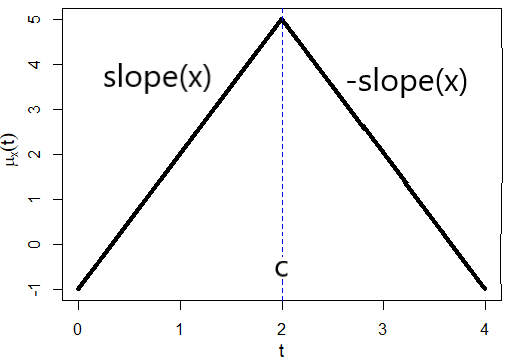}
\caption{Function $\mu_x(t)$ with parameters $c=2$ and $slope(x) =3$. }\label{Plot_slope}
\end{wrapfigure} 
\begin{itemize}
    \item Section~\ref{Section_simulations_high_dimension_confounder} provides insight into how our method scales with the dimension of the confounders $dim(\textbf{X})$.
     \item Section~\ref{section_sim1} extends the simple example presented in Section~\ref{Simulations_simple_example}, evaluating performance across different dependence structures (various copulas), sample sizes, and a range of causal effects.
    \item Section~\ref{Section_simulations_hidden_confounder} addresses a scenario involving a hidden confounder affecting both $T$ and $Y$.
    \item Section~\ref{section_simulations_extremal_region} focuses on variations in the function $\mu(t)$ and assesses the extent to which our method can extrapolate $\mu(t)$ into the 'extreme' region.
\end{itemize}

In some of the simulations, we use the following function: 
\begin{equation}\label{eq875}
    \mu_x(t) = \begin{cases}
      &  5-slope(x)(t-c)\,\,\,\,\,\,\,\,\,\,\,\,\,\,\,\,\,\,\text{ for }t\geq c,\\ 
     & 5+slope(x)(t-c)\,\,\,\,\,\,\,\,\,\,\,\,\,\,\,\,\,\,\text{ for }t<c,
    \end{cases}
\end{equation}
where typically $slope(x)=|x|$ and $c\in\mathbb{R}$ is a hyper-parameter. Graphical visualisation of the function $\mu_x$ for $x=
3$ can be found in Figure~\ref{Plot_slope}. In other words, $\mu_x$ grows with slope $x$ until $c$, and then declines with slope $x$.

\subsection{Simulations with a high dimensional $\textbf{X}$}
\label{Section_simulations_high_dimension_confounder}
In this simulations we consider $\textbf{X} = (X_1, \dots, X_d)$ where the dimension $d$ is potentially large. Consider the following data-generating process:
\begin{itemize}
    \item Let $a_1, \dots, a_d\overset{iid}{\sim} \mathcal{N}(1,1)$ and $b_1, \dots, b_d\overset{iid}{\sim} \mathcal{N}(-1,1)$ be fixed numbers at the beginning of the simulations. 
    \item Consider $\textbf{X}$ being centered Gaussian vector with $cor(X_i, X_j)=0.1$ for all $i\neq j$ and $var(X_i)=1$.
    \item Let $T = \sum_{i=1}^d a_iX_i + \varepsilon_T$, where $\varepsilon_T$ is distributed according to either $\mathcal{N}(0,10)$, $Exp(\frac{1}{10})$ or $Pareto(1,1)$. 
    \item Let $Y = \mu_{x}(T) +\sum_{i=1}^d b_iX_i + \varepsilon_Y $, where $\mu_x(t)$ is defined in (\ref{eq875}) with hyper-parameters $c=slope(x)=1$ and where $\varepsilon_Y\sim\mathcal{N}(0,1)$.
\end{itemize}

This data generating process leads to $\mu(t+1) - \mu(t) = -slope(x) = -1$ for $t\geq c$ and $ \mu(t+1) - \mu(t) = +1$ for $t<c$. Consequently, our primary interest lies in estimating $\omega = \mu(t+1) - \mu(t) =-1$ for $t\geq c$.

Note that a simple linear regression $Y\sim T + X_1 + \dots + X_d$ leads to a biased estimation of the effect of $T$, since the effect is different for large and for small values of $T$. However, simply discarding values where $T_i<1$ leads to a selection bias. 

We generate data as specified with a sample size of $n=5000$. Setting the threshold at $\tau = 0.95$, we employ the methodology outlined in Section~\ref{section_algorithm} to estimate $\omega$. Specifically, we utilize linear parametrization of the parameters in the estimation procedure. This process is repeated $100$ times. The mean of the estimates $\hat{\omega}$ together with $95\%$ quantile for various values of $d$ and distributions of the noise variables can be found in Table~\ref{table_simulations_dimension_d}.

Table~\ref{table_simulations_dimension_d} illustrates that with a sample size of $n=5000$, the results are reasonably accurate as long as $d\leq 50$. The reason for the bias observed in larger dimensions $d$ is that Assumption~\ref{assumption_max_domain} and (\ref{eq1}) are only asymptotic results, and with higher dimensions $d$, we require more data for the asymptotic theory for $T\mid \textbf{X}$ to take effect. It is well known that the convergence rate of the maxima of the Gaussian random sample to an extreme value distribution is very slow, whereas it is faster with Exponential or Pareto distribution \citep{Davis_converngence}. With a large dimension $d$, we also observe a more pronounced effect of the estimation error accumulated in the first step on the second step of the algorithm.

\subsection{Comparison with classical methods}\label{section_simulations_comaprison_classical_methods}
 
We evaluate our extrapolation method by comparing it with several state-of-the-art techniques from existing literature. Specifically, we assess the performance of four methods: the doubly robust estimation method introduced by Kennedy et al. \citep{Kennedy_2016, kennedy2019}, the additive spline estimator proposed by Bia et al. \citep{bia2014stata}, the approach suggested by Hirano and Imbens \citep{Hirano} employing a GAM outcome model (taken from \cite{causaldrf}), and the inverse probability of treatment weighting (IPTW) estimator by VanderWal et al. \citep{vanderwal2011ipw}.

Our analysis employs the same simulation setup described in Section~\ref{Section_simulations_high_dimension_confounder}, utilizing exponentially distributed noise variables (other distributions yield similar results). After generating $(\textbf{x}_i, t_i, y_i)_{i=1}^n$, we estimate $\mu(\tilde{t})$, where $\tilde{t} = \max_{i=1, \dots, n}(t_i)+10$, using all aforementioned methods. Subsequently, we compute the absolute relative error (ARE) defined as $$ARE = \bigg| \frac{\hat{\mu}(\tilde{t}) - \mu(\tilde{t})}{\mu(\tilde{t})}\bigg|.$$ This procedure is repeated 100 times, and the average of the obtained ARE values is presented in Table~\ref{table_classical_methods_result}.

Analysis of Table~\ref{table_classical_methods_result} reveals that our method achieves the smallest extrapolation error. Conversely, the IPTW method \citep{vanderwal2011ipw} exhibits poor performance due to the quadratic nature of the extrapolating curve. Kennedy et al.'s method \citep{Kennedy_2016, kennedy2019} typically produces a constant extrapolation curve. Meanwhile, both Bia et al.'s \citep{bia2014stata} and HI \citep{Hirano} methods are highly sensitive to extreme values in the dataset, leading to significant variability of the estimates.

 Please note that our method assumes linearity in the tail. If this assumption isn't met, our method may produce inferior results. In such instances, alternative regression techniques can be utilized instead of linear regression. For example, neural networks, as demonstrated in \cite{XinweiShen2023engression}, can be employed to address nonlinear behavior and enhance performance.

\subsection{Dependence, sample size and the causal effect}\label{section_sim1}
In the following, we conduct simulations based on a model with covariates $\textbf{X} = (X_1, X_2, X_3)$ that function as a common cause of both $T$ and $Y$. The details of the simulation are as follows:

\begin{itemize}
    \item $\textbf{X}$ is generated with standard Gaussian margins and a Gumbel copula with a parameter $\alpha$ (where $\alpha$ represents the degree of dependence \citep{kolesarova2018}; $\alpha=1$ corresponds to independence, and $\alpha\to\infty$ corresponds to full dependence; see Plot~\ref{Plot2} for an illustration with $\alpha = 2$).
    \item $T$ is generated in such a way that the marginal distribution of $T$ follows an exponential distribution with a scale parameter of 1, and the dependence structure between $\textbf{X}$ and $T$ follows a Gumbel copula with parameter $\alpha$.
    \item The response variable $Y$ is generated as follows:
    \begin{equation}
     Y = \begin{cases}
      & \frac{1}{2}\omega T+f(\textbf{X})+\varepsilon,\,\,\,\,\,\,\,\,\,\,\,\,\,\,\,\,\,\,\,\,\,\,\,\, \text{when }X_1>0, \,\,T>1,\\ 
     & \frac{3}{2}\omega T+f(\textbf{X})+\varepsilon,\,\,\,\,\,\,\,\,\,\,\,\,\,\,\,\,\,\,\,\,\,\,\,\, \text{when }X_1\leq0,\,\, T>1\\ 
     & -10T + 15 + f(\textbf{X})+\varepsilon,\,\,\,\,\,\,\,\,\,\, \text{when }T\leq 1,
    \end{cases}
    \end{equation}
    where $f$ is a randomly generated smooth function\footnote{To randomly generate a d-dimensional function, we use the concept of the Perlin noise generator \citep{PerlinNoise}. For more details, refer to the supplementary package. Readers can conceptualize this as a function ranging from quadratic to linear.}, $\varepsilon\sim N(0,1)$, and $\omega$ is a hyper-parameter that we vary in our simulations. Plot~\ref{Plot2} shows one realization of such a dataset. 
\end{itemize}

Please note that $\mu(t) = \omega t + \mathbb{E}f(\textbf{X})$ for any $t>1$. As such, our primary focus lies in estimating the slope $\mu(t+1) - \mu(t) = \omega$. We generate data with varying parameters, including $\omega$, $\alpha$, and the sample size $n$. Employing the method outlined in Section~\ref{section_algorithm}, we estimate $\omega$ across a spectrum of data-generating processes. For $n=1000$, we set the threshold at $\tau = 0.9$, and for $n>1000$, we use $\tau = 0.95$. This process is repeated 100 times, and the mean and $0.95\%$ quantile are presented in Table~\ref{table_SIM1}. The numbers in the brackets represent the mean of the estimated bootstrap confidence intervals. Ideally, these intervals should align with the $0.95\%$ quantile. 

The findings indicate that the methodology performs as anticipated in this simulation study: augmenting the sample size enhances the estimation, whereas elevating $\alpha$ (heightening the influence of the covariates) degrades the accuracy of the estimation. We observe that the bootstrap confidence intervals align relatively well with the actual $95\%$ quantiles.

\begin{figure}[h]
\centering
\includegraphics[scale=0.6]{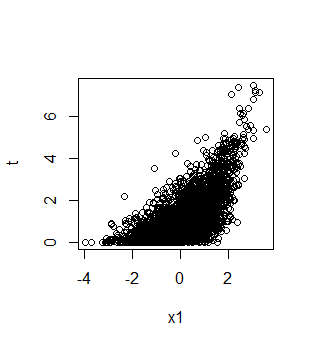}
\includegraphics[scale=0.45]{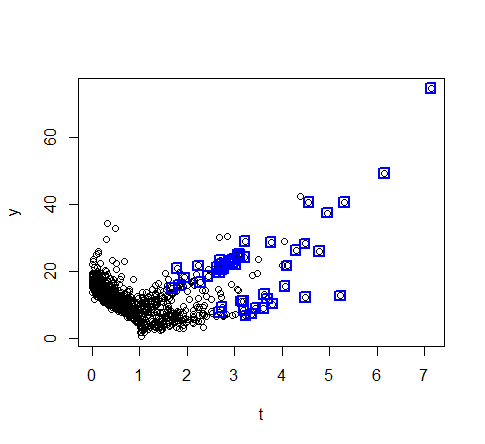}
\caption{The figures illustrate the dependence among $X_1$, $T$, and $Y$, generated based on the simulations outlined in Section~\ref{section_sim1} with a dependence parameter of $\alpha=2$ and $\omega=5$. Points falling within the set $S$ are identified by a blue square.}
\label{Plot2}
\end{figure}

\begin{table}[]
\begin{tabular}{cl|c|c|c|}
                                                  &           & $\omega=0$                             & $\omega=1$& $\omega=10$                            \\ \hline
\multicolumn{1}{c|}{\multirow{3}{*}{$\alpha = 1$}} & $n=1000$  & $\hat{\omega} = -0.05 \pm 0.33 (\pm 0.45)$& $\hat{\omega} = 0.94 \pm 0.60 (\pm 0.55)$& $\hat{\omega} = 9.86 \pm 2.99 (\pm 3.01)$\\
\multicolumn{1}{c|}{}                             & $n=5000$  & $\hat{\omega} = -0.01 \pm 0.28 (\pm 0.25)$& $\hat{\omega} = 0.98 \pm 0.29 (\pm 0.39)$& $\hat{\omega} = 10.09 \pm 1.91 (\pm 2.04)$\\
\multicolumn{1}{c|}{}                             & $n=10000$ & $\hat{\omega} = 0.00 \pm 0.18 (\pm 0.16)$& $\hat{\omega} = 0.99 \pm 0.21 (\pm 0.21)$& $\hat{\omega} = 9.95 \pm 1.62 (\pm 1.59)$\\ \hline
\multicolumn{1}{c|}{\multirow{3}{*}{$\alpha = 1.5$}} & $n=1000$  & $\hat{\omega} = 0.33 \pm 0.92 (\pm 0.96)$& $\hat{\omega} = 1.37 \pm 0.98 (\pm 1.13)$& $\hat{\omega} =  10.95 \pm 3.77 (\pm 3.38)$\\
\multicolumn{1}{c|}{}                             & $n=5000$  & $\hat{\omega} = 0.24 \pm 0.52 (\pm 0.68)$& $\hat{\omega} = 0.97 \pm 0.29 (\pm 0.25)$& $\hat{\omega} = 10.90 \pm 2.22 (\pm 2.43)$\\
\multicolumn{1}{c|}{}                             & $n=10000$ & $\hat{\omega} = 0.13 \pm 0.28 (\pm 0.49)$& $\hat{\omega} = 1.20 \pm 0.46 (\pm 0.55)$& $\hat{\omega} = 10.72 \pm 1.42 (\pm 1.76)$\\ \hline
\multicolumn{1}{c|}{\multirow{3}{*}{$\alpha = 2$}} & $n=1000$  & $\hat{\omega} = -0.17 \pm 1.19 (\pm 1.34)$& $\hat{\omega} = 0.99 \pm 1.01 (\pm 1.31)$& $\hat{\omega} =  11.14 \pm 3.32 (\pm 3.59)$\\
\multicolumn{1}{c|}{}                             & $n=5000$  & $\hat{\omega} = 0.03 \pm 0.66 (\pm 0.85)$& $\hat{\omega} = 1.05 \pm 1.01 (\pm 0.99)$& $\hat{\omega} = 11.15 \pm 2.91 (\pm 2.58)$\\
\multicolumn{1}{c|}{}                             & $n=10000$ & $\hat{\omega} = -0.09 \pm 0.50 (\pm 0.61)$& $\hat{\omega} = 0.96 \pm 0.59 (\pm 0.66)$& $\hat{\omega} = 10.70 \pm 2.02 (\pm 1.83)$\\ \hline
\end{tabular}
\caption{Resulting estimates of parameter $\omega = \mu(t+1) - \mu(t), t>1$ from Section~\ref{section_sim1}. Parameter $\alpha$ represents the dependence between $\textbf{X}, T$. The notation $\hat{\omega} = a \pm b (\pm c)$  represent the following: given  $100$ estimations of $\hat{\omega}$, $a$ is the mean, $b$ is the $95\%$ quantile and $c$ is the (average) $95\%$ quantile computed using bootstrap technique.  }
\label{table_SIM1}
\end{table}

\subsection{Simulations with a hidden confounder}
\label{Section_simulations_hidden_confounder}

Consider a similar simulations setup as in Section~\ref{Simulations_simple_example}, but with a hidden confounder. Consider an observed confounder $X = X_1 \sim \text{Bernoulli}(0.75)$ and a hidden confounder $H\sim N(1,1)$. Let $\delta\in\mathbb{R}$ and define $T = \delta H + X_1 + \varepsilon_T$, where $\varepsilon_T \sim \mathcal{N}(0,1)$. Note that $\delta$ represents the effect of a hidden confounder. Let
  \begin{equation}
     Y = \begin{cases}
      & \delta H+ \frac{2}{3}\omega T+\varepsilon,\,\,\,\,\,\,\,\,\,\,\,\,\,\,\,\,\,\,\,\,\,\,\,\, \text{when }X_1=1, \,\,T>1,\\ 
     & \delta H+ \frac{6}{3}\omega T+\varepsilon,\,\,\,\,\,\,\,\,\,\,\,\,\,\,\,\,\,\,\,\,\,\,\,\, \text{when }X_1=0,\,\, T>1,\\ 
     & \delta H + 3-2T  +\varepsilon,\,\,\,\,\,\,\,\,\,\, \text{when }T\leq 1,
    \end{cases}
    \end{equation}
where $\varepsilon \sim \mathcal{N}(0,1)$. A simple computation leads to $$\mu(t) = 0.75\frac{2}{3}\omega t + (1-0.75)\frac{6}{3}\omega t = \omega t$$ for any $t>1$, while $\mu(t) = -2t + 3$ for $t\leq 1$. Consequently, our primary interest lies in estimating $\omega$ for $t>1$. 

We generate data as specified with a sample size of $n$. Setting the threshold at $\tau = 0.9$, we employ the methodology outlined in Section~\ref{section_algorithm} to estimate $\omega$. This process is repeated $100$ times. The estimates $\hat{\omega}$ for a range of values of $\delta$ and $\omega$ and $n$ can be found in Table~\ref{Table_confounding}. 

The results in Table~\ref{Table_confounding} suggest that a hidden confounder does not bias our estimate unless its strength is very large. Indeed, Remark~\ref{remark1} suggest that Assumption~\ref{assumption_unconfoundness_tail} is still valid since the hidden confounder enters the equality in an additive way.

\begin{table}[]
\begin{tabular}{cl|c|c|c|}
                                                   &           & $\omega=0$                             & $\omega=5$                             & $\omega=10$                            \\ \hline
\multicolumn{1}{c|}{\multirow{3}{*}{$\delta= 0$}}  & $n=1000$  & $\hat{\omega} = 0.05\pm 0.22$& $\hat{\omega} = 5.03 \pm 0.31$& $\hat{\omega} = 10.02 \pm 0.4$\\
\multicolumn{1}{c|}{}                              & $n=5000$  & $\hat{\omega} = -0.01\pm 0.17$& $\hat{\omega} = 4.97 \pm 0.18$& $\hat{\omega} = 9.97\pm 0.20$\\
\multicolumn{1}{c|}{}                              & $n=10000$ & $\hat{\omega} = 0.01\pm 0.12$& $\hat{\omega} = 4.99 \pm 0.12$& $\hat{\omega} = 9.97 \pm 0.14$\\ \hline
\multicolumn{1}{c|}{\multirow{3}{*}{$\delta = 1$}} & $n=1000$  & $\hat{\omega} = 0.55 \pm 0.27$& $\hat{\omega} = 5.60 \pm 0.32$& $\hat{\omega} = 10.48\pm 0.33$\\
\multicolumn{1}{c|}{}                              & $n=5000$  & $\hat{\omega} = 0.53 \pm 0.12$& $\hat{\omega} = 5.50 \pm 0.18$& $\hat{\omega} = 10.51 \pm 0.18$\\
\multicolumn{1}{c|}{}                              & $n=10000$ & $\hat{\omega} = 0.50 \pm 0.08$& $\hat{\omega} = 5.48\pm 0.09$& $\hat{\omega} = 10.46\pm 0.16$\\ \hline
\multicolumn{1}{c|}{\multirow{3}{*}{$\delta = 5$}} & $n=1000$  & $\hat{\omega} = 0.96 \pm 0.04$& $\hat{\omega} = 5.98 \pm 0.22$& $\hat{\omega} = 10.91\pm 0.21$\\
\multicolumn{1}{c|}{}                              & $n=5000$  & $\hat{\omega} = 0.96 \pm 0.03$& $\hat{\omega} = 5.95 \pm 0.07$& $\hat{\omega} = 10.94\pm 0.16$\\
\multicolumn{1}{c|}{}                              & $n=10000$ & $\hat{\omega} = 0.96 \pm 0.025$& $\hat{\omega} = 5.94 \pm 0.04$& $\hat{\omega} = 10.92\pm 0.08$\\ \hline
\multicolumn{1}{c|}{\multirow{3}{*}{$\delta = 10$}} & $n=1000$  & $\hat{\omega} = 0.98 \pm 0.06$& $\hat{\omega} = 5.96 \pm 0.15$& $\hat{\omega} = 10.94\pm 0.30$\\
\multicolumn{1}{c|}{}                              & $n=5000$  & $\hat{\omega} = 0.99 \pm 0.025$& $\hat{\omega} = 5.98\pm 0.07$& $\hat{\omega} = 10.97\pm 0.14$\\
\multicolumn{1}{c|}{}                              & $n=10000$ & $\hat{\omega} = 0.99 \pm 0.02$& $\hat{\omega} = 5.97 \pm 0.04$& $\hat{\omega} = 10.95\pm 0.09$\\ \hline
\multicolumn{1}{c|}{\multirow{3}{*}{$\delta = 50$}} & $n=1000$  & $\hat{\omega} = 0.99 \pm 0.01$& $\hat{\omega} = 5.99 \pm 0.14$& $\hat{\omega} = 10.97\pm 0.28$\\
\multicolumn{1}{c|}{}                              & $n=5000$  & $\hat{\omega} = 0.99 \pm 0.01$& $\hat{\omega} = 5.99\pm 0.07$& $\hat{\omega} = 10.98\pm 0.15$\\
\multicolumn{1}{c|}{}                              & $n=10000$ & $\hat{\omega} = 1.0000 \pm 0.004$& $\hat{\omega} = 5.97 \pm 0.04$& $\hat{\omega} = 10.95\pm 0.08$\\ \hline
\end{tabular}
\caption{Resulting estimates of $\omega = \mu(t+1)-\mu(t), t>1$ from Section~\ref{Section_simulations_hidden_confounder}, together with $95\%$ quantile. Parameter $\delta$ represent the strength of a hidden confounder. }
\label{Table_confounding}
\end{table}

\subsection{Simulations with varying extremal region}\label{section_simulations_extremal_region}
In the following simulations, we explore variations in the function $\mu(t)$ and analyze the corresponding estimations $\hat{\mu}(t)$ for large values of $t$.

Consider the following data-generating process: 
\begin{align*}
&    X = \varepsilon_X, \quad \quad \quad \quad \quad\quad \quad  \varepsilon_X\sim N(0,1)\\&
     T = X + \varepsilon_T, \quad \quad \quad\quad \quad  \varepsilon_T\sim t_\nu \text{\footnotemark}, \\&
     Y = \mu_X(T) + \varepsilon_Y, \quad \quad \quad \varepsilon_Y\sim N(0,1),
\end{align*}
\footnotetext{Students t-distribution with $\nu$ degrees of freedom. Note that if $\nu\to\infty$ we obtain Gaussian distribution.}
where  $\mu_x(t)$ is defined in (\ref{eq875}). If $c$ is too large, we only observe the region where $\mu_x$ grows and hence, our estimation tends to be larger than the true value.

With varying $c$ and $\nu$, we estimate the parameter
\begin{align*}
    \omega = \mu(t+1) - \mu(t) = -\mathbb{E}|X| = -0.798\,\,\,\,\,\,\,\,\,\,\,\text{ for }t\geq c. 
\end{align*}

If we fit a linear model $\mathbb{E}Y = \beta_0 + \beta_TT+\beta_XX$, the estimate $\hat{\beta}_T$ tends to be positive (depending on $c$ and $\nu$, for example, if $c=\nu=5$, then $\hat{\beta}_T = 0.58\pm 0.02$). Using our methodology, we estimate $\hat{\omega}$ as in the previous simulations. The resulting numbers are presented in Table~\ref{table_slope_1}. We observe that if $c$ grows, our estimate becomes more biased as the data above the threshold still fall below $T<c$. Specifically, if $\nu=\infty$, only $0.2\%$ of data points have $T>5$, making the behavior of $\mu(t)$ above $t>c$ challenging to estimate. Note that the degrees of freedom $\nu$ correspond to the heavy-tailness of $T$; smaller $\nu$ values lead to more extreme values of $T$. Conversely, if $\nu=\infty$, $T$ follows a Gaussian distribution. Heavier tails of $T$ lead to better estimates.

\begin{table}[]
\begin{tabular}{|c|c|c|c|c|}
\hline
              True $\omega \approx -0.79$.& $c=1$& $c=2$& $c=5$& $c=10$\\ \hline
$\nu= \infty$& $\hat{\omega} = -0.75\pm 0.15$& $\hat{\omega} = -0.33\pm 0.13$& $\hat{\omega} = 0.63\pm 0.13$& $\hat{\omega} = 0.71\pm 0.38$\\ \hline
$\nu= 5$& $\hat{\omega} = -0.78\pm 0.18$& $\hat{\omega} = -0.72\pm 0.15$& $\hat{\omega} = -0.13\pm 0.15$& $\hat{\omega} = 0.6\pm 0.33$\\ \hline
$\nu= 2$& $\hat{\omega} = -0.77\pm 0.23$& $\hat{\omega} = -0.76\pm 0.22$& $\hat{\omega} = -0.7\pm 0.22$& $\hat{\omega} = -0.51\pm 0.19$\\\hline
\end{tabular}
\caption{Estimates of $\hat{\omega}$ with varying $c$ and $\nu$. Note that true $ \omega = -\mathbb{E}|X| \approx -0.79$ .}
\label{table_slope_1}
\end{table}

\section{Appendix: River data application}
\label{Appendix_application2}

\subsection{Simple illustration with known ground truth}

We used the following set of confounders:
\begin{itemize}
    \item $X_1=$Total precipitation (daily)
    \item $X_2=$Total precipitation during the previous 7 days
    \item $X_3=$Daily maximum of Air temperature 2 m above ground
    \item $X_4=$Daily maximum of Relative air humidity 2 m above ground
    \item $X_5=$Daily mean of Vapour pressure 2 m above ground
    \item $X_6=$Daily maximum of Pressure reduced to sea level
    \item $X_7=$Daily total of Reference evaporation from FAO
\end{itemize}

For pair $2\to 1$ we considered measurements from meteo-station M1 (station code MURI, AG) and for the remaining pairs we used measurements from M2 (a station with a code LUZ). However, variables $X_5, X_6, X_7$ were also measured in station M2 also for the pair $2\to 1$ since some values were missing and M2 has much longer time period of measurements. All these covariates can be safely considered as common causes of $T$ and $Y$, and no feedback loop is present. For modelling of $\theta(\textbf{X})$, we used linear parametrisation: that is, $\tau(\textbf{X}) = const +\sum_{i=1}^7\beta_{i, \tau}X_i$ and $\sigma(\textbf{X}) =  const + \sum_{i=1}^7\beta_{i, \sigma}X_i$ where the parameters $\beta_{i, \tau}$ were estimated using quantile regression (in \texttt{R} using \texttt{quantreg} package) and the parameters $\beta_{i, \sigma}$ were estimated using \texttt{evgam} package. We fixed $\xi(\textbf{X})$ constant. 

In the following, we focus on the pair $2\to 1$; for other pairs of stations, the results were similar. In the modelling of $\theta(\textbf{X})$,  $X_3$ and $X_5$ were not significant on $0.05$ level (note that for an estimation of $\omega$ we do not care which covariates were significant since the function $\theta(\textbf{X})$ is more or less unchanged and adding non-significant covariates only slightly increases variance of the estimation). Using non-parameteric GAM estimation of the parameters did not changed much the final estimation (from $\hat{\omega} = 1.03$ to $\hat{\omega} = 0.99$). Estimation of $\tau(\textbf{X})$ is also plotted in Figure~\ref{Plot3333}. Using linear model in this case does not seem to be very wrong, see plot~\ref{Plot333333} where except of the normality violation the model seem to fit quite well (in the body of the distribution).

\begin{figure}[h]
\centering
\includegraphics[scale=0.5]{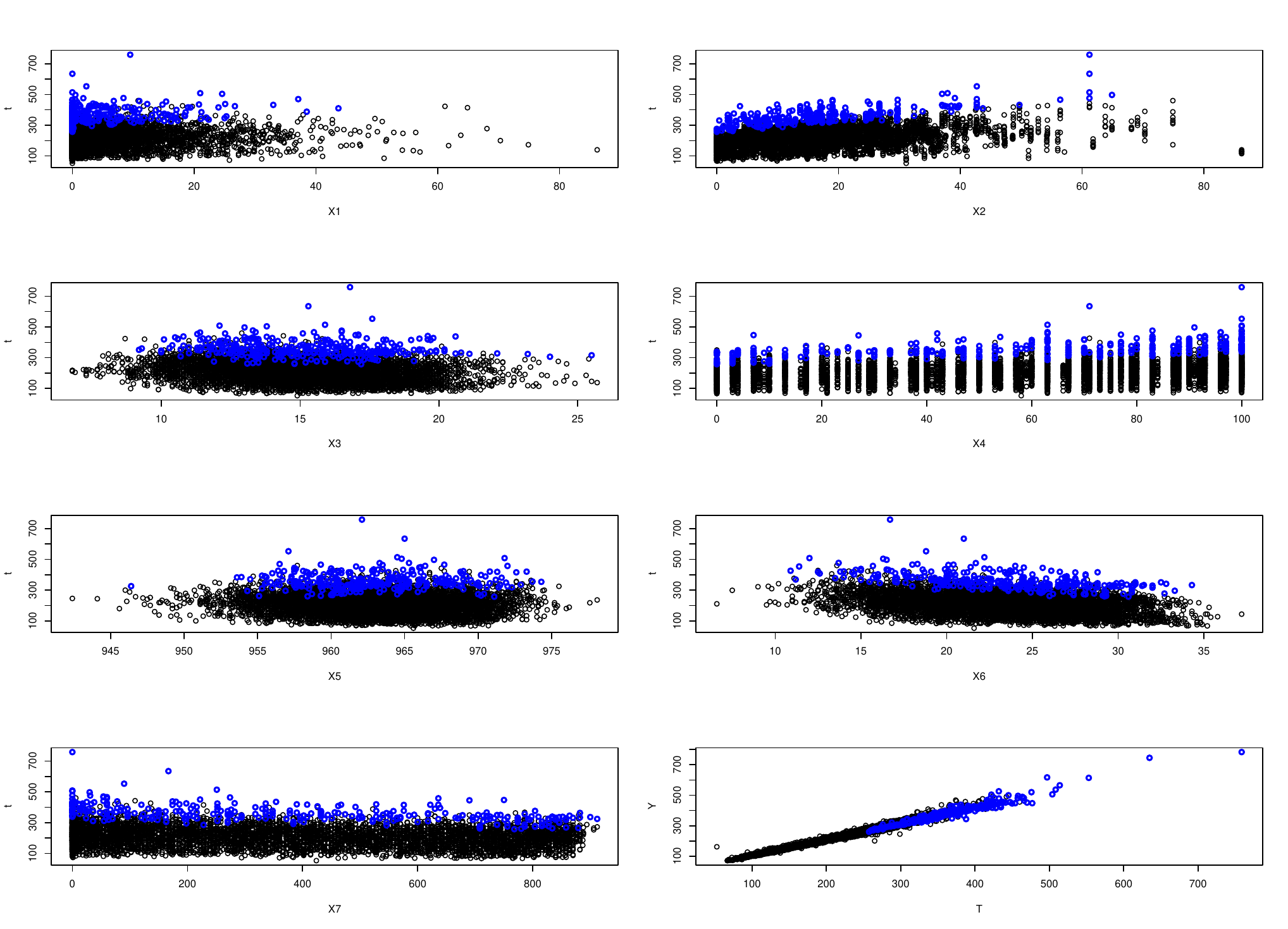}
\caption{ Visualisation of the estimation of $\tau(\textbf{x})$- estimated using classical quantile regression. Blue points characterise the observations above this threshold (points in the set $S$).   }
\label{Plot3333}
\end{figure}

\begin{figure}[h]
\centering
\includegraphics[scale=0.5]{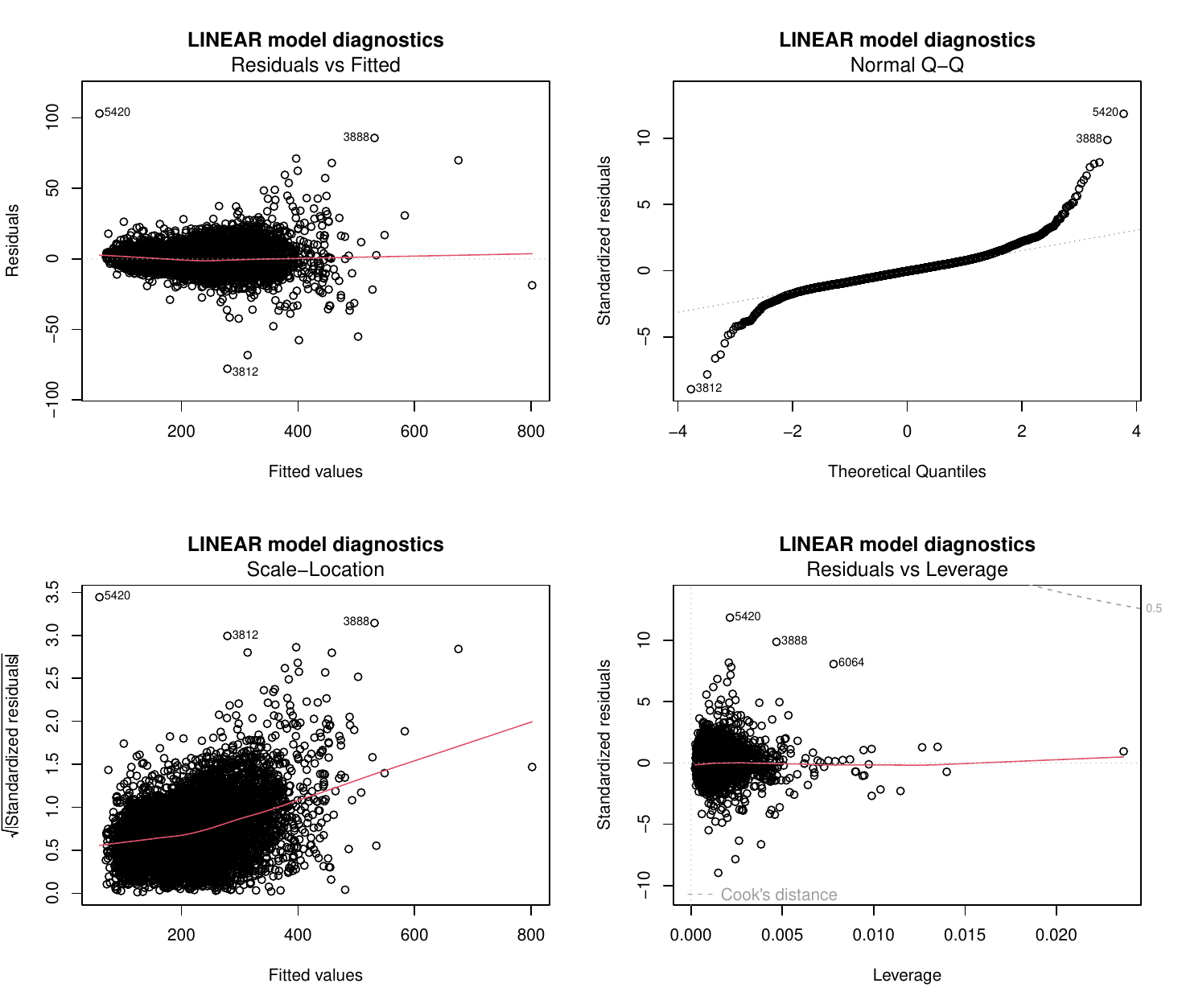}
\caption{ Diagnostics of a model $Y\sim X_1 + \dots + X_7$.}
\label{Plot333333}
\end{figure}

\subsection{Appendix: Effect of precipitation on river discharge}

Figure~\ref{Plot2222} visualize the dataset. 

\begin{figure}[h]
\centering
\includegraphics[scale=0.5]{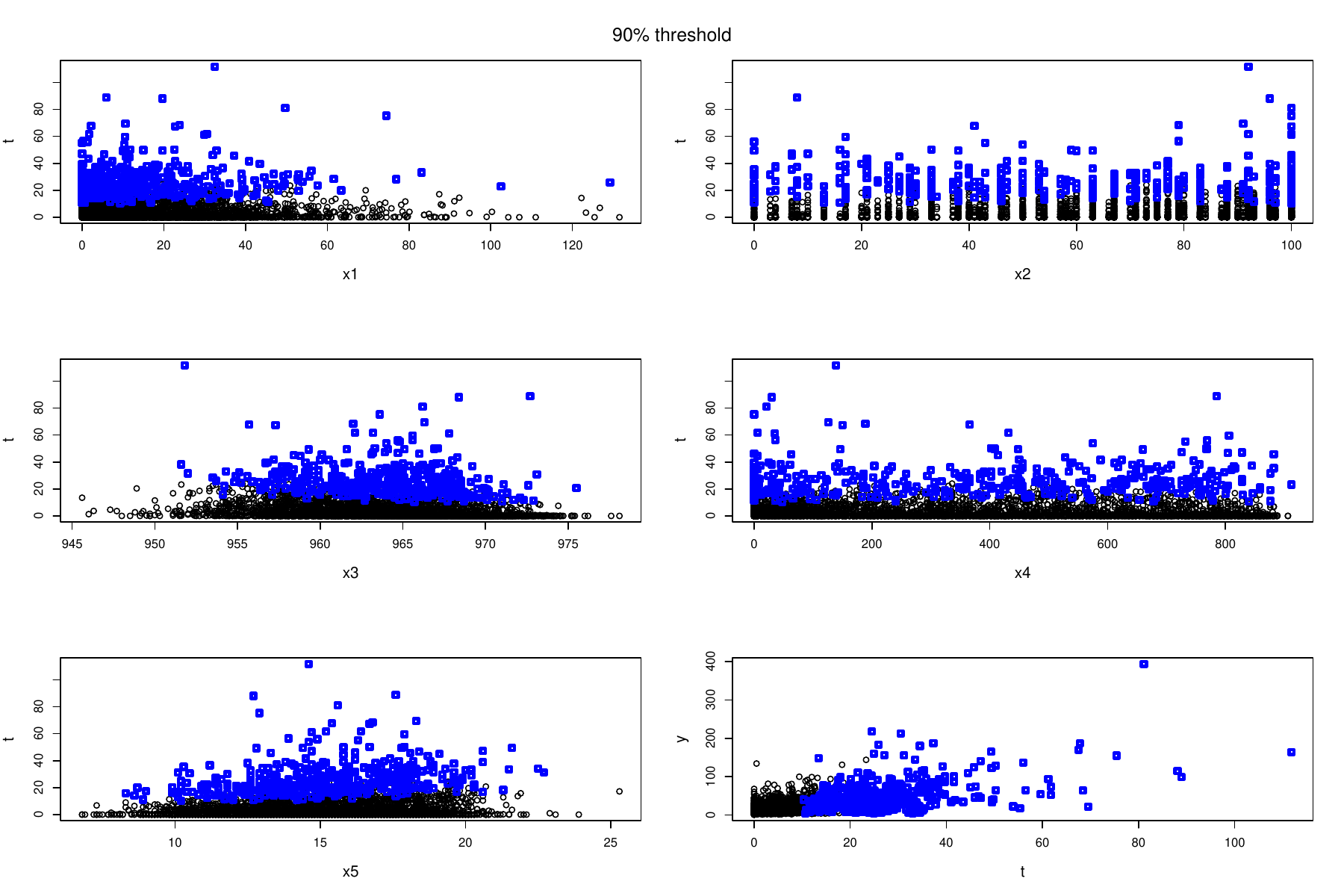}
\caption{ Visualisation of the estimation of $\tau(\textbf{x})$ for Part2 in the application- estimated using classical quantile regression. Blue points characterise the observations above this threshold (points in the set $S$).   }
\label{Plot2222}
\end{figure}

\subsubsection{Choice of variables}
We used the following set of variables:
\begin{itemize}
    \item $Y=$ River discharge on day $i+2$ 
    \item $T=$ Precipitation in the corresponding meteo-station on day $i+1$
    \item $X_1=$Total (sum) precipitation during the previous 7 days (days $i, i-1, \dots, i-6$)
    \item $X_2=$Daily maximum of Air temperature 2 m above ground on day $i$
    \item $X_3=$Daily maximum of Relative air humidity 2 m above ground on day $i$
    \item $X_4=$Daily maximum of Pressure reduced to sea level on day $i$
    \item $X_5=$Daily total of Reference evaporation from FAO on day $i$
\end{itemize}
where $i$ spans from 1.6.1930 up to 29.8.2014 (recall that we only considered the summer months). 
The choice of $Y$ and $T$ was addressed in the main text: since typically the auto-correlation peaking when $T$ represents the day prior to $Y$ (as illustrated in Figure~\ref{ccf}, as well as its extreme counterpart extremogram \citep{extremogram}). The rationale behind choosing $X_1$ is straightforward - precipitation over the preceding days emerges as a significant confounding factor affecting both $Y$ and $T$. Regarding additional variables, we opted for those deemed relevant and with reliable measurements across meteorological stations, all of which were recorded on the preceding day.

Is there a common cause between $Y$ and $T$ that remains unaccounted for? Variables $X_2$ through $X_5$, measured on day $i+1$, could serve as potential common causes for both $Y$ and $T$. For instance, a \textit{sudden} temperature change might elevate the likelihood of intense rainfall, while alterations in river discharge could stem from specific soil characteristics affected by the temperature change. Nevertheless, we contend that the majority of these variables require more than a day to manifest their effects, which we believe are largely encapsulated by our chosen variables.

\begin{figure}[h]
\centering
\includegraphics[scale=0.5]{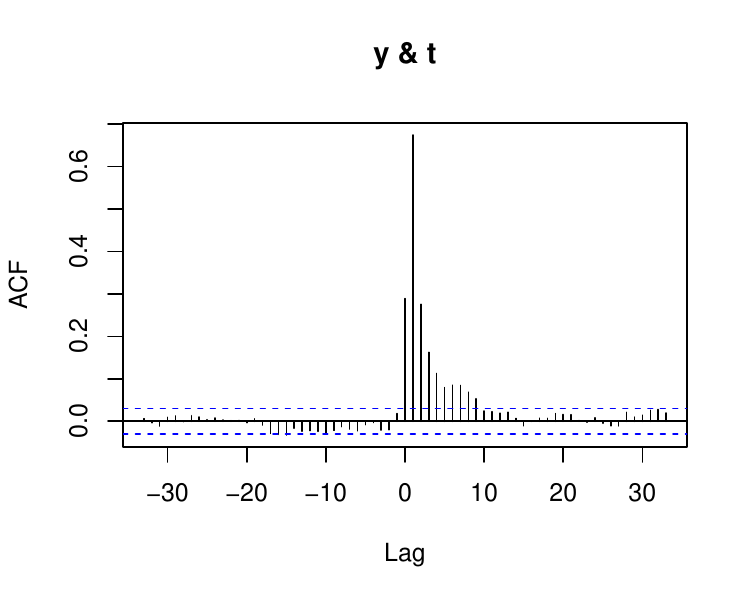}
\includegraphics[scale=0.5]{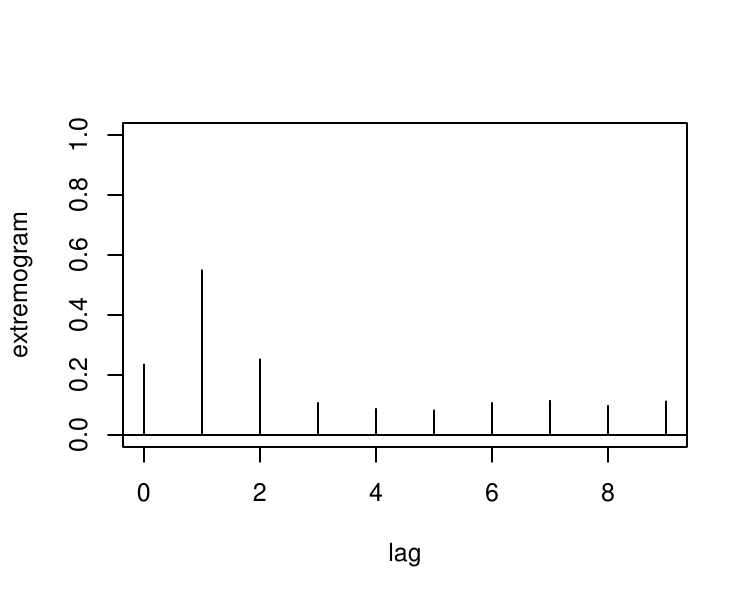}

\caption{Cross-correlation and cross-extremogram of precipitation recorded at meteostation M2 and water discharge at Station 3, both measured on the same day. }
\label{ccf}
\end{figure}

\subsubsection{Computation of $\hat{\beta}$ }
\label{application_river_computation_of_beta}
In Table \ref{table_stations}, we introduced a variable $\hat{\beta}$ that represent the effect of precipitation on the river discharge level in the body of $T$. This can be defined in several ways:
\begin{enumerate}
    \item Using the method introduced in \cite{Kennedy_2016, kennedy2019} we estimate $\hat{\mu}(t+1) - \hat{\mu}(t)$ for $ t = \mathbb{E}(T)$. 
    \item Using very straightforward approach where we model the data generating process of $Y$ using a linear structural equation model $Y = c+ \beta_T T + \beta_{\textbf{X}}\textbf{X} + \varepsilon$ and return the least square estimate of $\hat{\beta}_T$.
\end{enumerate}

Coincidentally, both approaches return a very similar value of $\hat{\beta}$ and hence, it does not matter which approach we use (values in Table~\ref{table_stations} are from using the second approach).

\section{Consistency, bootstrap and its asymptotics} \label{appendix_consistency}

In this section, we give more detailed description of the bootstrap algorithm and a more precise statement of Theorem~\ref{Theorem_consistency} together with its proof.  Theorem~\ref{theorem_consistency2} presents the consistency of $\hat{\mu}(t)$ for large $t$, while Theorem~\ref{theorem_consistency1} shows the consistency of $\hat{\omega}_{\textbf{x}^\star}$ under different assumptions. Note that, using the notation from Sections~\ref{section_algorithm} and \ref{section_simulations}, 
\begin{equation}\label{eq98798787}
    \begin{split}
 \hat{\omega}_{\textbf{x}^\star}&=\lim_{t\to\infty}\hat{\mu}_{\textbf{x}^\star}(t+1) - \hat{\mu}_{\textbf{x}^\star}(t)=\lim_{t\to\infty} \hat{\alpha}[\hat{\theta}(\textbf{x}^\star)]+ \hat{\beta}[\hat{\theta}(\textbf{x}^\star)](t+1) - \hat{\alpha}[\hat{\theta}(\textbf{x}^\star)]- \hat{\beta}[\hat{\theta}(\textbf{x}^\star)]t \\& =\hat{\beta}[\hat{\theta}(\textbf{x}^\star)].        
    \end{split}
\end{equation}

 \subsection{Bootstrap}\label{appendix_bootstrap_definition}
In what follows, we explain in details the procedure for an estimator $\hat{\zeta}_\alpha$ satisfying $$P(\omega_{\textbf{x}^\star}\leq  \hat{\zeta}_\alpha ) \geq  1-\alpha,\, \,\,\,\,\,\,\,\,\,\,\,\,\alpha\in(0,1).$$
We only focus on the upper confidence interval, the lower and both-sided intervals can be done analogously. Our approach is standard and \cite{Vaart_1998} provides a good overview. 

Let $P_n$ be the empirical distribution of the observations $\textbf{Z}_i:=(\textbf{X}_i, T_i, Y_i), i=1, \dots, n$. We draw a random sample   $(\textbf{Z}^\star_{1}, \dots \textbf{Z}^\star_{n})\overset{iid}{\sim} P_n$, and we compute the parameter $\hat{\omega}^\star$ from $(\textbf{Z}^\star_{1}, \dots \textbf{Z}^\star_{n})$ the same way as we compute $\hat{\omega}$ from $(\textbf{Z}_1, \dots, \textbf{Z}_n)$. We define $\hat{\zeta}_\alpha$ as the upper $\alpha$-quantile of $\hat{\omega}^\star$: that is, the smallest value $x = \hat{\zeta}_\alpha$ that satisfies 
$$
P\bigg(\hat{\omega}^\star - \hat{\omega}\leq x \mid P_n \bigg)\geq 1-\alpha.
$$
The notation \(P(\cdot | P_n)\) indicates that the distribution of $\hat{\omega}^\star$ must be evaluated assuming that the observations are sampled according to \(P_n\) given the original observations. In particular, in the preceding display \(\hat{\omega}\) is to be considered nonrandom.

It is almost never possible to calculate the bootstrap quantiles exactly \citep{Vaart_1998}. In practice, these estimators are approximated by a simulation procedure. A large number of independent bootstrap samples \(\textbf{Z}_1^\star, \ldots, \textbf{Z}_n^\star\) are generated according to the estimated distribution \(P_n\). Each sample gives rise to a bootstrap value $\hat{\omega}^\star$. Finally, the bootstrap quantiles are estimated by the empirical quantiles of these bootstrap values. This simulation scheme always produces an additional (random) error in the coverage probability of the resulting confidence interval. In principle, this error can be made arbitrarily small by using a sufficiently large number of bootstrap samples. Therefore, the additional error is usually ignored in the theory of the bootstrap procedure. This section follows this custom and concerns the "exact" quantiles, without taking a simulation error into account.

\subsection{Simplifying assumptions}\label{asdadsfgw}

We simplify some steps in the inference process in order to simplify the proof of the consistency. In particular, we assume the following:

\begin{enumerate}
  \renewcommand{\labelenumi}{\textbf{\Alph{enumi})}}
    \item  \textbf{(Causality justification)} Consider Assumptions  \ref{assumption_max_domain}, \ref{assumption_unconfoundness_tail}  and \ref{assumption_linearity_of_tail_1} to be valid.
    \item \textbf{(Step 2 convergence)}  $\mathbb{E}Y^2<\infty$, $\mathbb{E}||\textbf{X}||^2<\infty$ and $(\textbf{X}, T)$ satisfy Grenander conditions \footnote{This is a minor assumption assuring that the matrix of observations have a full rank with probability tending to one. See Table 4.2 in \cite{greene2008econometric}.}. 
    \item \textbf{(Step 1 convergence)} We assume that conditions R1, R2 and R3 from \cite{chernozhukov2005extremal} are satisfied. That is, $E(\textbf{XX}^\top)$ is positive semi-definite, $\textbf{X}$ has a compact support with existing and finite quantile densities $\frac{\partial F^{-1}_U(\tau|x)}{\partial \tau}$, $\frac{\partial F^{-1}_U(\tau)}{\partial \tau}$ where $U = T-\tau_{lin}^\top \textbf{X}$,  $\tau_{lin}\in\mathbb{R}^d$.
    \item \textbf{(Linearity)} Assume that functions $\theta$, $\alpha, \beta$ are linear, functions $\sigma, \xi$ are constant and that we employ linear regression for the estimation of the parameters. 
%That is, we assume that for a large $q$ holds $Q_T(q\mid \textbf{X}) = \tau(\textbf{X}) = \tau_{lin}^\top \textbf{X}$ where $Q_T(q\mid \textbf{X})$ is a conditional $q$-quantile of $T$ given $\textbf{X}$. 
\end{enumerate}

In particular, following the notation in Section~\ref{section_algorithm} and using the notation $\tau(\textbf{x}) = \tau_{lin}^\top\textbf{x}$, our algorithm is as follows: 

\begin{itemize}
    \item choose $q\in(0,1)$,
    \item \textbf{(Step 1)} estimate $\tau_{lin}\in\mathbb{R}^d$ by minimising $\hat{\tau}_{lin} \in argmin_b\sum_{i=1}^nh_q(T_i - \textbf{X}_i^\top b)$ where $h_q(x) = x (q 1_{x\geq 0} - (1-q)1_{q<0})$.  
    \item \textbf{(Step 2)} we estimate $\alpha, \beta\in\mathbb{R}$ using least squares in a model      \begin{equation}
        \label{asfe}
        \mathbb{E}[Y\mid T=t, \tau(\textbf{X})=\hat{\tau}_{lin}^\top\textbf{x}] = \alpha\hat{\tau}_{lin}^\top\textbf{x} + \beta\hat{\tau}_{lin}^\top\textbf{x}t,    \end{equation}  from the data-points in $S$ (that is, we only consider $t>\hat{\tau}_{lin}^\top\textbf{x}$). Using  $\textbf{\textsf{R}}$ language we run the following code:
     \texttt{ fit = lm(Y $\sim$ s +s:T$_S$, data = data.frame(s, T$_S$))}, where $s = \hat{\tau}_{lin}^\top\textbf{X}_S$, $T_S = \{T_i: i\in S\}$ and $\textbf{X}_S = \{\textbf{X}_i: i\in S\}$. 
    \item we output $\hat{\omega}_{\textbf{x}^\star} = \hat{\beta}\hat{\tau}_{lin}^\top\textbf{x}^\star$ (see (\ref{eq98798787})). 
\end{itemize}

\begin{remark}
    Assumption~C implies consistency of $\hat{\tau}_{lin}$ (under the assumption that $q$ is chosen as a function of the sample size $n$, denoted as $q=q_n$ satisfying $\lim_{n\to\infty}q_n=1$ and $\lim_{n\to\infty}n(1-q_n)=\infty$), see Theorem~5.1 in \cite{chernozhukov2005extremal}. This assumption can be simplified by directly assuming consistency of $\hat{\tau}_{lin}$.  
\end{remark}
\subsection{Consistency }

We present two consistency results. One concerns the consistency of $\hat{\mu}(t)$ under general non-linear assumptions but under a neglecting the GPD approximation error. The second result describes the consistency of  $\hat{\omega}_{\textbf{x}^\star} $ under linear assumptions presented in Section~\ref{asdadsfgw}.

\begin{theorem}[Consistency] \label{theorem_consistency2}
Consider Assumptions  \ref{assumption_max_domain}, \ref{assumption_unconfoundness_tail}  and \ref{assumption_linearity_of_tail_1} to be valid. 
\begin{itemize}
    \item Assume that $\theta$, $\alpha$, and $\beta$ are continuous functions, and suppose we employ consistent estimators for $\theta$, $\alpha$, and $\beta$. For instance, the Generalized Additive Model (GAM) estimator \citep{Wood} has been shown to be consistent under specific smoothness conditions.
    \item  Let $q\in (0,1)$ be chosen such that the distribution of $T\mid T>\tau_q(\textbf{x}), \textbf{X}=\textbf{x}$ follows $ GPD(\theta(\textbf{x}))$ for all $\textbf{x}\in \mathcal{X}$, where $\mathcal{X}=supp(\textbf{X}) $ is assumed to be compact. 

\end{itemize}

Under these conditions, our estimator is consistent in the sense that for all $t\in\mathcal{T}$
$$
\hat{\mu}(t) \overset{P}{\to} \tilde{\mu}(t)\,\,\,\,\,\,\,\,\,\,\,\,\,\,\,\text{as } n\to\infty,
$$
where $\tilde{\mu}$ is a function that satisfies $\tilde{\mu}(t)\sim \mu(t)$ as $t\to\tau_R$. 
\end{theorem}

The second assumption outlined in Theorem~\ref{theorem_consistency2} is introduced to address certain technical hurdles that arise when dealing with a quantile $q$ which varies with the sample size $n$. Broadly speaking, when $q$ is not fixed, the statistical framework becomes considerably more intricate, making the task of demonstrating the consistency of a quantile regression notably challenging. Please note that while the distribution of $T\mid T>\tau_q(\textbf{x}), \textbf{X}=\textbf{x}$ converges to a Generalized Pareto Distribution ($GPD(\theta(\textbf{x}))$) in the limit for large $q$, exact validity is limited to special cases, such as when $T\mid \textbf{X}$ follows a Pareto distribution. However, by selecting $q$ sufficiently large, one can mitigate this issue, effectively reducing the disparity between the distributions of $T\mid T>\tau_q(\textbf{x}), \textbf{X}=\textbf{x}$ and $GPD(\theta(\textbf{x}))$ to insignificance. This theorem, therefore, provides valuable insight into the general consistency of the model, despite the idealized nature of the assumption.

Note that Theorem~\ref{theorem_consistency2} can be reformulated for $\mu_{\textbf{X}}$ analogously.

The subsequent theorem does not necessitate a fixed $q$; however, it presupposes linearity in the models for $T$ and $Y$. 
\begin{theorem}[Consistency] \label{theorem_consistency1}
Under Assumptions~A,B,C and D, where $q$ is chosen as a function of the sample size $n$, denoted as $q=q_n$ satisfying $\lim_{n\to\infty}q_n=1$ and $\lim_{n\to\infty}n(1-q_n)=\infty$, our estimator $\hat{\omega}_{\textbf{x}^\star}$ is consistent. That is, 
$$\hat{\omega}_{\textbf{x}^\star} - {\omega}_{\textbf{x}^\star}\overset{P}{\to} 0, \quad \text{as} \quad n\to\infty. $$
\end{theorem}

\begin{proof}[Proof of Theorem~\ref{theorem_consistency2}]
The proof is very straightforward. Lemma~\ref{lemma2} shows that $$\mu(t) \sim \int_{\mathcal{X}} \mathbb{E}[Y\mid  T=t, \theta(\textbf{x})]p_{\theta(\textbf{X})}(\textbf{x})d\textbf{x}\,\,\,\,\,\,\,\,\, \text{   for    }t\to\tau_R.$$
Assumption~\ref{assumption_linearity_of_tail_1} allows us to rewrite (correctness of this step follows directly from Lemma~\ref{lemma_o_integraloch} by considering $f(t,x)= \mathbb{E}[Y\mid  T=t, \theta(\textbf{x})]$ and $g(t,x) = \alpha(\theta(\textbf{x})) + \beta(\theta(\textbf{x}))t$)
$$\int_{\mathcal{X}} \mathbb{E}[Y\mid  T=t, \theta(\textbf{x})]p_{\theta(\textbf{X})}(\textbf{x})d\textbf{x} \sim  \int_{\mathcal{X}} [\alpha(\theta(\textbf{x})) + \beta(\theta(\textbf{x}))t ]\,\,\, p_{\theta(\textbf{X})}(\textbf{x})d\textbf{x}:=\tilde{\mu}(t)\,\,\,\,\,\,\,\,\, \text{   for    }t\to\tau_R.$$
Since $\theta$, $\alpha$, and $\beta$ are continuous and their estimators are consistent, we get that for all $t\in\mathcal{T}$ holds 
$$
\int_{\mathcal{X}} [\hat{\alpha}(\hat{\theta}(\textbf{x})) + \hat{\beta}(\hat{\theta}(\textbf{x}))t ]\,\,\, p_{{\theta}(\textbf{X})}(\textbf{x})d\textbf{x} \overset{P}{\to} \tilde{\mu}(t) \,\,\,\,\,\,\,\text{as } n\to\infty.
$$
Moreover, from the law of large numbers, it holds that $$\hat{\mu}(t) - \int_{\mathcal{X}} [\hat{\alpha}(\hat{\theta}(\textbf{x})) + \hat{\beta}(\hat{\theta}(\textbf{x}))t ]\,\,\, p_{{\theta}(\textbf{X})}(\textbf{x})d\textbf{x} \overset{P}{\to} 0, \,\,\,\,\,\,\,\text{as } n\to\infty.$$ 
Together, we obtain $$\hat{\mu}(t) \overset{P}{\to} \tilde{\mu}(t), \,\,\,\,\,\,\,\text{as } n\to\infty, $$
where the function on the right side is tail-equivalent with $\mu(t)$, what we wanted to show. 
\end{proof}

\begin{lemma}
\label{lemma_o_integraloch}
Let $\mathcal{X}$ be compact set and $\tau_R$ be the right endpoint of $\mathcal{T}\subseteq\mathbb{R}$. Let $f,g: \mathcal{T}\times \mathcal{X}\to\mathbb{R}$  be continuous functions such that for all $x\in\mathcal{X}$ holds $f(t,x)\sim g(t,x)$ as $t\to\tau_R$. Let $F$ be a continuous distribution function. Then, 
$$
\int_{\mathcal{X}} f(t,x) dF(x) \sim \int_{\mathcal{X}} g(t,x) dF(x), \,\,\,\,\,\,as\,\,\,t\to\tau_R.
$$
    
\end{lemma}

\begin{proof}
Let $\varepsilon>0$. Find $t_0$ such that for all $t>t_0$ and for all $x\in\mathcal{X}$ holds $1-\varepsilon < \frac{f(t,x)}{g(t,x)}<1+\varepsilon$. Then for any $t>t_0$ holds
$$
\frac{\int_{\mathcal{X}} f(t,x) dF(x)}{\int_{\mathcal{X}} g(t,x) dF(x)} < \frac{\int_{\mathcal{X}} f(t,x) dF(x)}{\int_{\mathcal{X}} \frac{1}{1+\varepsilon}f(t,x) dF(x)} = 1+\varepsilon
$$
and analogously $\frac{\int_{\mathcal{X}} f(t,x) dF(x)}{\int_{\mathcal{X}} g(t,x) dF(x)} >1-\varepsilon$. Proof is finished by sending $\varepsilon\to 0$. 
\end{proof}

\begin{proof}[Proof of Theorem~\ref{theorem_consistency1}]
\textbf{Idea: } \textit{We assume that the GPD approximation and linear model approximations are correct up to a factor of $\varepsilon$: we argue that for a large $n$ this is correct. Next, we use Theorem 5.1 in \cite{chernozhukov2005extremal} to show consistency of $\hat{\tau}$. We use Theorem 4.4 in \cite{greene2008econometric} to show consistency of $\hat{\beta}$ together with linearity of the least square estimate (to show that it does not depend on the inaccuracy of the estimate of $\hat{\tau}$). Finally, we use Lemma~\ref{lemma2} and send $\varepsilon\to 0$. }

\textbf{Proof: }Let $\varepsilon>0$. We claim that is possible to find $t<\tau_R$ and $n_0\in\mathbb{N}$ such that for all $\textbf{x}\in\mathcal{X}$ and all $n\geq n_0$, the following five statements hold with arbitrarily large probability: 
\begin{itemize}
    \item $t<F^{-1}_{T\mid \textbf{X} = \textbf{x}}(q_{n})$ (In other words, $q_{n_0}$ is large enough such that the  $q_{n}$-quantile of $T\mid\textbf{X}$ is larger than $t$. )
    \item It holds that $$1-\varepsilon<\frac{\mu_{\textbf{x}^\star}(t)}{\mathbb{E}[Y\mid  T=t, \tau(\textbf{X})=\tau(\textbf{x}^\star)]}<1+\varepsilon,$$where $\tau(\textbf{x}) = \tau_{lin}^\top \textbf{x}$ is the $q_{n_0}$-quantile of $T\mid \textbf{X}=x$,
    \item It holds that $$1-\varepsilon<\frac{\mathbb{E}[Y\mid T=t, \tau(\textbf{X})=\tau(\textbf{x}^\star)]}{ \alpha{\tau}(\textbf{x}^\star) + \beta{\tau}(\textbf{x}^\star)t}<1+\varepsilon,$$
    \item  $||\hat{\tau}_{lin} - \tau_{lin}||<\varepsilon$,
    \item $|\hat{\beta} - \beta|<\varepsilon$, 
\end{itemize}

where $\hat{\tau}_{lin} = argmin_b\sum_{i=1}^{n_0}h_{q_{n_0}}(T_i - \textbf{X}_i^\top b)$ is the maximum likelihood estimator, and $\beta$ is the real coefficient in the model 
\begin{equation}\label{yukyuk}
    \mathbb{E}[Y\mid T=t, T>{\tau}_{lin}^\top\textbf{x}, \textbf{X}=\textbf{x}] = \alpha{\tau}_{lin}^\top\textbf{x} + \beta{\tau}_{lin}^\top\textbf{x}t,
\end{equation}
and $\hat{\beta}$ is the corresponding least square estimate. 

We prove the bullet-points here:
\begin{itemize}
    \item First bullet-point is a trivial consequence of the assumption  $q_n\to 1$.
    \item Second bullet-point is a trivial consequence of Lemma~\ref{lemma2} together with Assumption~D,
    \item Third bullet-point is a trivial consequence of Assumptions~\ref{assumption_linearity_of_tail_1} and D,
    \item The fourth bullet-point follows from a well-known consistency of $\hat{\tau}_{lin}$. It is well known that for a fixed quantile $q$, the maximum likelihood estimator $\hat{\tau}_{lin} = argmin_b\sum_{i=1}^nh_q(T_i - \textbf{X}_i^\top b)$ is consistent and even asymptotically normal (see e.g. Theorem 4.1 in \cite{koenker_2005}, noting that we assume continuous $T$ and finite second moments of $\textbf{X}$). However, quantile $q$ is not fixed and is increasing with the sample size with the speed  $\lim_{n\to\infty}q_n=1$ and $\lim_{n\to\infty}n(1-q_n)=\infty$. This is a well-known generalization of quantile regression known as 'intermediate order regression quantiles' or 'moderately extreme quantiles' \citep{chernozhukov2016extremal} and is as well consistent and asymptotically normal under Assumption~C (see Theorem 5.1 in \cite{chernozhukov2005extremal}). 
    \item The fifth bullet-point: For a moment, fix $\tau_{lin}\neq \textbf{0}$. It is an elementary knowledge that the estimation of $\beta$ using least squares in a model (\ref{yukyuk}),
  where $\tau_{lin}$ is fixed, is consistent and even asymptotically normal under conditions $var(Y)<\infty$ ,  $\mathbb{E}||\textbf{X}||^2<\infty$, $(\textbf{X}, T$) satisfying Grenander conditions and the sample-size $|S|=:k_n = n(1-q_n)\to\infty$.    (see e.g. Lemma~\ref{lemma_grenander}). Observe that least squares estimate $\hat{\beta}$ is linear in $\tau_{lin}$: that is, if we express $\hat{\beta}$ explicitly, we get $\hat{\beta} = \tau_{lin}^\top \widehat{\tilde{\beta}}$, where   $\tilde{\beta}$ is a coefficient in a linear model corresponding to (\ref{ewqfdwef}) (where $T$ is assumed to be larger than $\tau^\top X$ implicitly). 
Finally, using this observation, we can replace the fixed value of  $\tau_{lin}$ by a random $\hat{\tau}_{lin}$, and we still obtain $\hat{\beta} = \hat{\tau}_{lin}^\top \widehat{\tilde{\beta}}$. Since by increasing $n$ we can make $\widehat{\tilde{\beta}}$ arbitrarily accurate with arbitrarily large probability, the same holds for  $\hat{\beta}$.   In the following paragraph, we present an an illustration of the linearity of  $\hat{\beta}$  in $\tau_{lin}$ for $d=1$.  An explicit expression of $\hat{\beta}$ as a function of $\tau_{lin}$ and our data explicitly is the following: 
\begin{equation}
\begin{bmatrix}
    \hat{\alpha} \\
    \hat{ \beta}
\end{bmatrix}
=
(M^\top M)^{-1}M^\top\textbf{Y}_S, \,\,\,\,\,\,\,\,where\,\,\,\,\, 
M=\begin{bmatrix}
    {\tau}_{lin}\textbf{x}_{1} &  {\tau}_{lin}\textbf{x}_{1}t_{1} \\ \dots &  \dots \\ {\tau}_{lin}\textbf{x}_{k}&  {\tau}_{lin}\textbf{x}_{k}t_{k} \\
\end{bmatrix}
\end{equation}
where $\textbf{Y}_S= (Y_1, \dots, Y_k)^\top$, WLOG let $S=\{1, \dots, k_n\}\subset\{1, \dots, n\}$.   Note that \begin{equation}
M =  {\tau}_{lin} diag(x_1, \dots, x_k)
\begin{bmatrix}
 1 & t_{1} \\ \dots &  \dots \\ 1&  t_{k} \\
\end{bmatrix}={\tau}_{lin} \tilde{M},
\end{equation}
where  $\tilde{M}$ is the data-matrix corresponding to a model (\ref{ewqfdwef}). 
\end{itemize}

Combining all the bullet-points, we obtain that with arbitrarily large probability that
\begin{equation*}
    \begin{split}
 \mu_{\textbf{x}^\star}(t) &\approx\mathbb{E}[Y\mid  T=t, \tau(\textbf{X})=\tau(\textbf{x})]\\&\approx \alpha{\tau}(\textbf{x}) + \beta{\tau}(\textbf{x})t\\&\approx \hat{\alpha}{\hat{\tau}_{lin}}^\top\textbf{x}^\star + \hat{\beta}{\hat{\tau}_{lin}}^\top\textbf{x}^\star t,       
    \end{split}
\end{equation*}
where each sign '$\approx$' represent equality up to a factor of $\varepsilon$ (in either multiplicative or additive form) which is negligible as $\varepsilon\to 0$. This implies consistency, Quod erat demonstrandum. 

\begin{lemma}\label{lemma_grenander}
Consider an estimate $(\hat{\alpha}, \hat{\tilde{\beta}})$ of $\alpha\in\mathbb{R}, \tilde{\beta}\in\mathbb{R}^d$ using least squares in a model 
\begin{equation}\label{ewqfdwef}
    \mathbb{E}[Y\mid T=t, \textbf{X}=\textbf{x}] = \alpha\textbf{x} + \tilde{\beta}^\top\textbf{x}t,
\end{equation}
based on a random sample $ (Y_1, T_1, X_1), \dots (Y_k, T_k, X_k)$. Then, $\hat{\tilde{\beta}}$ is consistent and asymptotically normal if $var(Y)<\infty$  and $\mathbb{E}||\textbf{X}||^2<\infty$ and $(\textbf{X}, T$) satisfy Grenander conditions.  

Proof can be found in Theorem 4.4 in \cite{greene2008econometric}. 
\end{lemma}

\end{proof}

\subsection{Bootstraps correctness }

We use results from Chapter 23.2 in \cite{Vaart_1998}. The main step is to use delta-method for bootstrap (Theorem~23.5 in \cite{Vaart_1998}) and the fact that regression models in step 1 and step 2 of our algorithm are 'bootstrappable' (Theorem 3 in \cite{hahn1995bootstrapping} and \cite{freedman1981bootstrapping}). 

To simplify some steps of the proof, we assume the following:

\begin{enumerate}[label=\textbf{\Alph*}.]
    \setcounter{enumi}{4} % Start from the fifth letter (E)
    \item Assume that $\hat{\omega}_{\textbf{x}^\star}$ is consistent (which holds for example under assumptions A,B,C,D).
    \item We compute $\hat{\tau}_{lin}$ from the first $\lfloor \frac{n}{2} \rfloor$ data-points and we compute $\hat{\beta}$ from the remaining $\lceil \frac{n}{2} \rceil$ data-points.
    \item In the computation of the set $S$, we assume that $\tau$ is known and non-random; that is,   $S = \{i\leq n: T_i>\tau(\textbf{X}_i)\}$ instead of $S = \{i\leq n: T_i>\hat{\tau}(\textbf{X}_i)\}$.
    \item Assumption of Theorem 3 in \cite{hahn1995bootstrapping} are satisfied; that is, $\mathbb{E}[\textbf{X}\textbf{X}^\top]$ is non-singular matrix, the conditional density of $Y - \tau^\top\textbf{X}$ given $\textbf{X}$, denoted as $f$, satisfies $f(\epsilon \mid \textbf{X}) > r_1$ whenever $|\epsilon|\leq r_2$ for some positive numbers $r_1$, $r_2$. Finally, there exists some function $G$ such that $f(\epsilon \mid \textbf{X}) \leq G(\textbf{X})$ for all $\epsilon$ and $\mathbb{E}\left[(1 + G(\textbf{X})) ||\textbf{X} ||^2\right] < \infty.$

\end{enumerate}

\begin{theorem}
\label{theorem_bootstrap_consistency}
Assume validity of assumptions D,E,F,G and H. Let $q\in (0,1)$ be chosen such that the distribution of  $T\mid T>\tau_q(\textbf{x}), \textbf{X}=\textbf{x}$ follows $ GPD(\theta(\textbf{x}))$ for all $\textbf{x}\in\mathcal{X}$. Then, $\hat{\zeta}_\alpha$ is asymptotically consistent; that is, $$liminf_{n\to\infty}P(\omega_{\textbf{x}^\star}\leq \hat{\zeta}_\alpha) \geq 1-\alpha.$$
\end{theorem}

\begin{proof}
    We will show that both  $\frac{1}{\sqrt{n}}(\hat{\omega}_{\textbf{x}^\star} - {\omega}_{\textbf{x}^\star})$ and   $\frac{1}{\sqrt{n}}(\hat{\omega}^\star_{\textbf{x}^\star} - \hat{\omega}_{\textbf{x}^\star})$  given $P_n$ both converge to the same distribution, say $G$. That is,   $\frac{1}{\sqrt{n}}(\hat{\omega}_{\textbf{x}^\star} - {\omega}_{\textbf{x}^\star})\overset{D}{\to}G\overset{D}{\leftarrow}\frac{1}{\sqrt{n}}(\hat{\omega}^\star_{\textbf{x}^\star} - \hat{\omega}_{\textbf{x}^\star})$ as $n\to\infty$. This directly implies (see e.g. Lemma 23.3 in \cite{Vaart_1998}) that  $\hat{\zeta}_\alpha$ is asymptotically consistent.

\textbf{Observation 1) } $ \hat{\tau}_{lin}$ satisfies that $\frac{1}{\sqrt{n}}(\hat{\tau}_{lin}- {\tau}_{lin})$ and   $\frac{1}{\sqrt{n}}(\hat{\tau}_{lin}^\star - \hat{\tau}_{lin})$  given $P_n$ both converge to the same Gaussian distribution  (see Theorem 3 in \cite{hahn1995bootstrapping}). 

\textbf{Observation 2)}  $\hat{\beta} = \hat{\tau}_{lin}^\top \widehat{\tilde{\beta}}$, where $\tilde{\beta}$ is a coefficient in a linear model corresponding to (\ref{ewqfdwef}) (where $T$ is assumed to be larger than $\tau(\textbf{X})$ implicitly since we assumed that $\tau$ is known and non-random in $S$). Note that  $\widehat{\tilde{\beta}}\indep \hat{\tau}_{lin}$ since $\widehat{\tilde{\beta}}$ is computed from the second half of the dataset and its computation does not contain $\hat{\tau}_{lin}$. However, we know that $\widehat{\tilde{\beta}}$ satisfies that $\frac{1}{\sqrt{n}}(\widehat{\tilde{\beta}}- {\tilde{\beta}})$ and   $\frac{1}{\sqrt{n}}(\widehat{\tilde{\beta}}^\star - \widehat{\tilde{\beta}})$  given $P_n$ both converge to the same Gaussian distribution (Theorem 2 in \cite{ECK2018141} or \cite{freedman1981bootstrapping}).

\textbf{Together:} since $ \hat{\tau}_{lin}$ is 'bootstrappable' and $\widehat{\tilde{\beta}}$ is 'bootstrappable' and they are independent, the delta method give us that $\hat{\omega}_{\textbf{x}^\star}$ is 'bootstrappable'. More formally, we use Theorem~23.5 in \cite{Vaart_1998} (Delta method for bootstrap). Define  $\phi: \mathbb{R}^{2d}\to\mathbb{R}: \phi(a,b) = (a^\top b) (a^\top \textbf{x}^\star)$. Note that $\hat{\omega}_{\textbf{x}^\star} = \phi(\hat{\tau}_{lin}, \widehat{\tilde{\beta}})$. Since  $ \hat{\tau}_{lin}$ and $\widehat{\tilde{\beta}}$ satisfy the conditions of the theorem, we get that  $\frac{1}{\sqrt{n}}(\hat{\omega}_{\textbf{x}^\star} - {\omega}_{\textbf{x}^\star})$ and   $\frac{1}{\sqrt{n}}(\hat{\omega}^\star_{\textbf{x}^\star} - \hat{\omega}_{\textbf{x}^\star})$  given $P_n$ both converge to the same distribution. That is what we wanted to show. 

\end{proof}

\section{Appendix: Proofs of Lemma \ref{lemma1} and Lemma \ref{lemma2}}
\label{Section_proofs}

\textbf{Proof of Lemma~\ref{lemma1}}:
A simple computation gives us   \begin{equation*}
    \begin{split}
        \mathbb{E}[Y(t)] 
    &= \mathbb{E}(\mathbb{E}[Y(t)\mid \textbf{X}]) \sim  \mathbb{E}(\mathbb{E}[Y\mid \textbf{X}, T=t])           
   \\&=\int_{\mathcal{X}}\int_{\mathcal{Y}} p_{Y\mid \textbf{X}, T}(y, \textbf{x}, t)p_\textbf{X}(\textbf{x})y\,\,\,\,dyd\textbf{x}
   \\&=\int_{\mathcal{X}}\int_{\mathcal{Y}} \frac{p_T(t)}{p_{T\mid \textbf{X}}(t\mid \textbf{x})} p_{Y, \textbf{X}\mid T}(y,\textbf{x}\mid t) y\,\,\,\,\,dy d\textbf{x} 
   \\&=\mathbb{E}\{\pi_0(T, \textbf{X})Y\mid T=t\}.
    \end{split}
\end{equation*}

\textbf{Proof of Lemma~\ref{lemma2}}:
From Assumption~\ref{assumption_unconfoundness_tail}, we have that 
\begin{equation}\label{eq987}
    \mathbb{E}[Y(t)\mid  \textbf{X}]\sim \mathbb{E}[Y\mid  T=t, \textbf{X}]\,\,\,\,\,\,\,\,\, \text{   as    }t\to\tau_R.
\end{equation}
On both sides of (\ref{eq987}), we condition on $\theta(\textbf{X})$ and integrate over the remaining $\textbf{X}$ (denoted as $\theta^C(\textbf{X})$, formally it is an orthogonal complement). Note that the distribution of $T \mid \theta(\textbf{X})$ approaches the distribution of $T \mid \textbf{X}$, given $T > \tau(\textbf{X})$ for sufficiently large $\tau(\textbf{X})$, since it approaches $GPD(\theta(\textbf{X}))$.  Hence, $P_{\theta^C(\textbf{X})\mid T=t,\theta(\textbf{X})}$ approaches the distribution $P_{\theta^C(\textbf{X})\mid \theta(\textbf{X})}$ as $t\to\tau_R$. 

We obtain the following:
 \begin{align*}
\mathbb E[Y(t)\mid \theta(\textbf{X})]&=\int\mathbb E[Y(t)\mid \theta(\textbf{X}),\theta^C(\textbf{X})=w]dP_{\theta^C(\textbf{X})\mid\theta(\textbf{X})}(w)\\&
\sim\int\mathbb E[Y(t)\mid \theta(\textbf{X}),\theta^C(\textbf{X})=w]dP_{\theta^C(\textbf{X})\mid T=t,\theta(\textbf{X})}(w)\\&
=\int\mathbb E[Y\mid T=t,\theta(\textbf{X}),\theta^C(\textbf{X})=w]dP_{\theta^C(\textbf{X})\mid T=t,\theta(\textbf{X})}(w)\\&
=\mathbb E[Y\mid T=t,\theta(\textbf{X})].
\end{align*}

The second statement in the Lemma trivially follows from the first by integrating over $\theta(\textbf{X})$.

%If you came here because you want your references in a new page, uncomment the following line

\clearpage % If you want the references in a separate page
\bibliography{bibliography}

\end{document}